\documentclass[%
 11pt,
 tightenlines,
nofootinbib,
 amsmath,amssymb,
 aps,
 pra
]{revtex4-2}

\usepackage{graphicx}
\usepackage{xcolor}
\usepackage{amsthm}
\usepackage[margin=1in]{geometry}
\usepackage[pdftex,colorlinks=true,linkcolor=red,citecolor=red,urlcolor=black]{hyperref}
\usepackage{comment}
\usepackage{multirow}

\def\a{{\bf a}}
\def\b{{\bf b}}
\def\c{{\bf c}}

\def\x{{\bf x}}
\def\y{{\bf y}}

\def\T{{\tau}}
\def\VEC{{\rm vec}}

\renewcommand{\d}{\mathrm{d}}

\newcommand{\ket}[1]{| #1 \rangle}

\newcommand{\proj}[1]{| #1 \rangle \langle #1 |}

\DeclareMathOperator{\poly}{poly}

\newcommand{\be}{\begin{equation}}
\newcommand{\ee}{\end{equation}}
\def\ba{\begin{array}}
\def\ea{\end{array}}
\newcommand{\bes}{\begin{equation*}}
\newcommand{\ees}{\end{equation*}}
\newcommand{\bea}{\begin{eqnarray}}
\newcommand{\eea}{\end{eqnarray}}
\newcommand{\beas}{\begin{eqnarray*}}
\newcommand{\eeas}{\end{eqnarray*}}
\def\ba{\begin{array}}
\def\ea{\end{array}}

\makeatletter
\newtheorem*{rep@theorem}{\rep@title}
\newcommand{\newreptheorem}[2]{%
\newenvironment{rep#1}[1]{%
 \def\rep@title{#2 \ref{##1} (restated)}%
 \begin{rep@theorem}}%
 {\end{rep@theorem}}}
\makeatother

 \newtheorem{prob}{Problem}

\newtheorem{thm}{Theorem}
\newtheorem*{thm*}{Theorem}
\newtheorem{cor}[thm]{Corollary}

\newtheorem{lem}[thm]{Lemma}
\newtheorem*{lem*}{Lemma}

\newtheorem{prop}[thm]{Proposition}
\newtheorem{defn}[thm]{Definition}
\newtheorem{rem}[thm]{Remark}

\newtheorem{fact}[thm]{Fact}

\newreptheorem{thm}{Theorem}
\newreptheorem{lem}{Lemma}
\newreptheorem{cor}{Corollary}
\newreptheorem{prob}{Problem}
\newreptheorem{prop}{Proposition}

\allowdisplaybreaks

\usepackage{blindtext}
\usepackage{scrextend}

\usepackage{algorithm}
\usepackage{algorithmic}

\makeatletter
\newenvironment{breakablealgorithm}
  {
   \begin{center}
     \refstepcounter{algorithm}
     \hrule height1pt depth0pt \kern3pt
     \renewcommand{\caption}[2][\relax]{
       {\raggedright\textbf{\ALG@name~\thealgorithm} ##2\par}%
       \ifx\relax##1\relax 
         \addcontentsline{loa}{algorithm}{\protect\numberline{\thealgorithm}##2}%
       \else 
         \addcontentsline{loa}{algorithm}{\protect\numberline{\thealgorithm}##1}%
       \fi
       \kern3pt\hrule\kern3pt
     }
  }{
     \kern3pt\hrule\relax%
   \end{center}
  }
\makeatother

\begin{document}

\title{Solving generalized eigenvalue problems by ordinary differential equations on a quantum computer}

\author{Changpeng Shao}
\email{changpeng.shao@bristol.ac.uk}
\affiliation{\footnotesize{School of Mathematics, University of Bristol, BS8 1UG, UK}}

\author{Jin-Peng Liu}
\email{jliu1219@terpmail.umd.edu}
\affiliation{\footnotesize{Joint Center for Quantum Information and Computer Science, University of Maryland, MD 20742, USA \\
Institute for Advanced Computer Studies, University of Maryland, MD 20742, USA \\
Department of Mathematics, University of Maryland, MD 20742, USA}}

\date{\today}

\begin{abstract}

Many eigenvalue problems arising in practice are often of the generalized form $A\x=\lambda B\x$. One particularly important case is symmetric, namely $A, B$ are Hermitian and $B$ is positive definite. The standard algorithm for solving this class of eigenvalue problems is to reduce them to Hermitian eigenvalue problems. For a quantum computer, quantum phase estimation is a useful technique to solve  Hermitian eigenvalue problems. In this work, we propose a new quantum algorithm for symmetric generalized eigenvalue problems using ordinary differential equations. The algorithm has lower complexity than the standard one based on quantum phase estimation. Moreover, it works for a wider case than symmetric: $B$ is invertible, $B^{-1}A$ is diagonalizable and all the eigenvalues are real.

\end{abstract}

\maketitle



\setlength{\parskip}{3pt}


\section{Introduction}

Quantum phase estimation (QPE), which is one of the most useful techniques in designing quantum algorithms, is a quantum algorithm to estimate the eigenvalues of  unitary matrices \cite{kitaev1995quantum}.   It is frequently used as a subroutine in many quantum algorithms, e.g., quantum algorithm for finding eigenvalues of Hermitian matrices \cite{abrams1999quantum}, quantum counting \cite{brassard1998quantum}, quantum linear solver \cite{harrow2009quantum},  quantum singular value estimation \cite{kerenidis_et_al:LIPIcs:2017:8154}, and so on (e.g., \cite{abrams1999quantum,lloyd2014quantum,rebentrost2014quantum,ambainis2016efficient,parker2020quantum,brassard2002quantum}).  Taking the Hermitian eigenvalue problem as an example. For any $n\times n$ Hertimian matrix $H$ with eigenpairs $\{(\lambda_j,\ket{E_j}):j=1,\ldots,n\}$, the QPE returns $\sum\beta_j \ket{\tilde{\lambda}_j} \ket{E_j}$ when the input state is $\ket{\phi} = \sum\beta_j \ket{E_j}$, where $\tilde{\lambda}_j$ is an approximation of $\lambda_j$. We can think of the output state as a quantum version of eigenvalue decomposition of $H$. Performing measurements on the first register, we obtain the approximations of the eigenvalues. Actually, the output state can have more applications than estimating eigenvalues. For example, we can use it to create a state proportional to $H^{-1}\ket{\phi}$. This gives a quantum algorithm for solving linear systems \cite{harrow2009quantum}.
In QPE, the input state $\ket{\phi}$ can be any desired state, and there is no need to know the explicit decomposition $\ket{\phi} = \sum\beta_j \ket{E_j}$ (e.g., a famous example is the Shor's algorithm \cite{shor1999polynomial}, see the analysis in \cite{NielsenChuang}). This simple fact makes QPE more flexible to solve other problems.

Hermitian eigenvalue problems  is an important class of eigenvalue problems.  However, many eigenvalue problems arising in applications (e.g., linear discriminant analysis \cite{james2013introduction}, canonical-correlation analysis \cite{hardoon2004canonical}, homogeneous Fredholm equation of the second type \cite{Baker}, constrained least squares problem \cite{tisseur2001quadratic}, robust eigenvector classifiers \cite{mangasarian2005multisurface}, etc.) are not of the standard form $A\x=\lambda \x$ but of the generalized form $A\x=\lambda B\x$. An important case is the symmetric generalized eigenvalue problem (GEP), i.e., $A,B$ are Hermitian and $B$ is positive definite. The standard algorithm for solving symmetric GEP is to reduce it to a standard Hermitian eigenvalue problem $(B^{-1/2} A B^{-1/2}) B^{1/2}\x = \lambda B^{1/2}\x$. Consequently, we can apply QPE to solve it  directly in a quantum computer. 
For example, Parker and Joseph recently studied symmetric GEPs through this idea for some types of $A,B$ arising from applications in physics such that $B^{-1/2} A B^{-1/2}$ is sparse and can be constructed efficiently \cite{parker2020quantum}. 

In this paper, we propose a new quantum algorithm for symmetric GEPs that demonstrates higher efficiency than the standard algorithm based on QPE. In addition, our algorithm works for a wider range of situations than symmetric and improves the previous quantum algorithm for the standard eigenvalue problem \cite{shao2019computing}.

\subsection{Problem setting}

In this work, we focus on the solving of the following problem, which can be viewed as a generalization of the Hermitian eigenvalue problem solved by QPE.

\begin{prob}[Quantum version of GEP (QGEP)]
\label{prob:problem 1}
Let $\epsilon\in(0,1)$ be the error tolerance.
Let $A,B$ be two $n\times n$ complex matrices with $B$ non-singular, $B^{-1}A$ diagonalizable. 
Denote the eigenpairs of the GEP $A\x=\lambda B\x$ as
$\{(\lambda_j, \ket{E_j}):j=1,\ldots,n\}$. Suppose all the eigenvalues are real.
Given access to copies of the state $|\phi\rangle$, which formally equals $\sum_j \beta_j|E_j\rangle$, the goal is to output a state proportional to $ \sum_j \beta_j|\tilde{\lambda}_j\rangle|E_j\rangle$, where $|\lambda_j - \tilde{\lambda}_j|\leq \epsilon$ for all $j$.
\end{prob}

In QGEP, we assumed that $B$ is non-singular and $B^{-1}A$ is diagonalizable.
In the following, we give the reasons for this.
\begin{enumerate}
 \item When $B$ is singular, the GEP may demonstrate some unusual properties than the standard eigenvalue problem. For instance, it may happen that the GEP has infinitely many eigenvalues or $\lambda = \infty$ is an eigenvalue (see Section \ref{section:Preliminaries} for more details). Also, when $B$ is singular, the GEP is well known to be ill-conditioned \cite{hochstenbach2019solving}. So it is hard to solve even for a quantum computer. Indeed, we will give a quantum lower bound analysis for this in Section \ref{section:ow bounds for the generalized eigenvalue problem}. The result indicates that to solve a singular GEP, a quantum computer needs at least to make $\Omega(\sqrt{n})$ queries to $B$. So we will not consider singular GEPs in this paper.
 \item The assumption of diagonalizability implies that any vector in theory can be uniquely decomposed into a linear combination of the eigenvectors. So the initial state $\ket{\phi}$ of QGEP can be any desired state. Equivalently, there is no restriction on the choice of the input state. Since the eigenvectors are unknown to us, we do not have the decomposition in advance. Thus to solve the QGEP in practice we need the quantum algorithm to be independent of this decomposition.
\end{enumerate}

There is a close connection between Problem \ref{prob:problem 1} and symmetric GEP.
On one hand, For symmetric GEP $A\x = \lambda B \x$, it is known that there is an invertible matrix $E$ such that $AE=BE\Lambda$ with $\Lambda$ real diagonal \cite{Golub} (see Proposition \ref{lemma:real generalized eigenvalues} for a proof). So symmetric GEP is a special case of Problem \ref{prob:problem 1}. 
On the other hand, Problem \ref{prob:problem 1} can be reduced to a symmetric GEP. More precisely, as shown in \cite{drazin1962criteria}, a diagonalizable matrix $M$ that only have real eigenvalues if and only if there are two Hermitian matrices $H_1, H_2$ with $H_1$ or $H_2$ positive definite such that $M=H_1H_2$. So in Problem \ref{prob:problem 1}, $B^{-1}A = H_1H_2$ for a Hermitian pair $(H_1,H_2)$. Suppose $H_1$ is positive definite, then $H_2 \x = \lambda H_1^{-1} \x$ defines a symmetric GEP. It has the same eigenpairs as the original GEP $A\x=\lambda B\x$. If $H_2$ is positive definite,  we can focus on $(H_2^{1/2} H_1 H_2^{1/2}) H_2^{1/2} \x =\lambda H_2^{1/2} \x$. Usually, $H_1,H_2$ are not easy to compute, so it may not be straightforward to solve Problem \ref{prob:problem 1} by the QPE in this way. However, our quantum algorithm given in this paper is independent of this kind of decomposition.

A more general case than symmetric such that all the eigenvalues are real is called definite. A Hermitian matrix pair $(A,B)$ is called definite if 
$
\min_{\x \in \mathbb{C}^n,\|\x\|_2=1} 
(\x^\dagger A \x)^2
+(\x^\dagger B \x)^2 > 0.
$
As shown in \cite{stewart1990matrix}, if $(A,B)$ is a definite pair, there is a $\theta\in[0,2\pi)$ such that $A\cos\theta - B \sin\theta, A\cos\theta + B \sin\theta$ are Hermitian and $A\cos\theta + B \sin\theta$ is positive definite. So all the eigenvalues of a definite pair are real. A potential problem of this reduction is the calculation of $\theta$. It is usually not easy to find such a $\theta$, so the standard algorithm based on QPE may not be efficient for all definite pairs through this idea. While our algorithm proposed in this paper works for all definite pairs that satisfy the assumption of Problem \ref{prob:problem 1}, and there is no need to find $\theta$.

We remark that in practice, we are  usually  more concerned about some maximal or minimal eigenvalues. So solving Problem \ref{prob:problem 1} may not be the most effective way to compute the extreme eigenvalues. For those particular eigenvalue problems, other algorithms (e.g., power method, Lanczos method, variational quantum eigensolver) should be more powerful. The advantage of Problem \ref{prob:problem 1} is that it contains the whole information of the eigenpairs in one quantum state so that measuring the first register returns all the eigenvalues up to some additive error. Also, the output state can be viewed as a quantum version of eigenvalue decomposition, which can have more applications than estimating eigenvalues.

\subsection{Our result}

In this paper, we will design a quantum algorithm for Problem \ref{prob:problem 1} in the  framework of block-encoding. Block-encoding is an effective framework to manipulate matrix operations in a quantum computer. There are also many efficient methods to construct block-encodings of  matrices. So quantum algorithms building on this framework should be general enough to solve problems related to  matrices. 

Let $M$ be a matrix and $U$ be a unitary, we say $U$ a block-encoding of $M$ if there is a  positive real number $\alpha \in \mathbb{R}^+$ such that
$$
U =
\left[ \begin{array}{ccccccc}
M/\alpha & \cdot \\
\cdot & \cdot \\
\end{array} \right].
$$
When having a block-encoding, many fundamental operations (e.g., matrix multiplication, matrix inversion, the singular value decomposition, etc) on $M$ become effective in a quantum computer \cite{chakraborty2018power, Gilyen-QSVT}. Below, we shall call $U$ an $\alpha$-block-encoding of $M$ for simplicity. The rigorous terminology is given in Section \ref{subsec:Block-encoding}.

Denote the matrix of the generalized eigenvectors as $E$, the condition numbers of $A,B,E$ as $\kappa_A,\kappa_B,\kappa_E$ respectively. Since the GEP is invariant under scaling, we can do a scaling to make sure that $\|B\|= \Theta(1)$. This assumption is just for the convenience of notation below, our algorithm is indeed independent of this assumption. For simplicity, we assume that the block-encodings of $A,B$ are constructed efficiently in time polylog in $n$. Then our main result can be stated as the follows:

\begin{thm}[Informal of Theorem \ref{thm:GEP by ODE-general}]
Given an $\alpha_A$-block-encoding of $A$ and
an $\alpha_B$-block-encoding of $B$, then Problem \ref{prob:problem 1} can be solved in time
\be \label{intro:complexity1}
\widetilde{O}\big(
\kappa_E(\alpha_A + \rho \alpha_B ) \kappa_B/\epsilon
\big),
\ee
where $\rho$ is an upper bound on the eigenvalues.
\end{thm}

An obvious fact is that for any $\alpha$-block encoding $U$ of a matrix $M$, we have $\alpha \geq \|M\|$. Also in the quantum case, we hope $U$ is efficiently implementable in the quantum circuits. So it may happen that $\alpha \gg \|M\|$. For instance, we may have $\alpha=\|M\|_F$ assuming access to the QRAM data \cite{chakraborty2018power}. Here $\|M\|_F$ is the Frobenius norm. Because of this, we keep $\alpha_A, \alpha_B$ in the complexity (\ref{intro:complexity1}).
Note that $\rho \leq \|B^{-1}A\| = O(\kappa_B\|A\|)$ under the assumption of $\|B\|= \Theta(1)$. So in the worst case, the complexity of our algorithm is upper bounded by 
$\widetilde{O}\big(
 \kappa_E(\alpha_A + \|A\|\alpha_B\kappa_B )
\kappa_B/\epsilon
\big)$.

The algorithm based on QPE for symmetric GEPs is not comprehensively studied in the past.\footnote{A plausible reason: the QPE idea for symmetric GEP is so naive that it is not deserve to write a single paper. In this paper, we use the best known block-encoding techniques to do the complexity analysis of this method.}   For comparison, we will also analyze its complexity in this paper (see Theorem \ref{thm:GEP by block-encoding}).
The complexity of QPE is determined by the cost of constructing the block-encoding of $B^{-1/2}AB^{-1/2}$. We will use two methods to construct this block-encoding. The overall complexity of the first method is
\be \label{intro:complexity2}
\widetilde{O}\big(
\alpha_A\alpha_B\kappa_B^{2.5} /\epsilon 
\big),
\ee
and of the second method is
\be \label{intro:complexity3}
\widetilde{O}\big(
(\alpha_A+\|A\|\kappa_B\alpha_B)\kappa_B^{2} /\epsilon 
\big).
\ee
The second method has better dependence on $\alpha_A,\alpha_B$, but worse dependence on  $\kappa_B$ than the first method. 
When restricted to the symmetric case, $\kappa_E = \sqrt{\kappa_B}$ (see Proposition \ref{lemma:real generalized eigenvalues}). So the complexity (\ref{intro:complexity1}) of our algorithm is at most
\be
\widetilde{O}\big(
(\alpha_A + \|A\|\kappa_B \alpha_B ) \kappa_B^{1.5}
/\epsilon
\big).
\ee
Therefore, our algorithm has better dependence on $\alpha_A,\alpha_B$ and $\kappa_B$ than the algorithm based on QPE. At the end of Section \ref{section:Block-encoding method}, we will discuss more about the connection, which involves some other notation, between our algorithm and the standard algorithm based on QPE. In summary, it states that our algorithm combines the advantages of the two above methods based on QPE. If we further restricted to the Hermitian eigenvalue problem (i.e., $A$ is Hermitian, $B=I$), then all the algorithms have the same complexity $\widetilde{O}(\alpha_A/\epsilon)$. This is consistent with the complexity of QPE  \cite{chakraborty2018power}.

\subsection{Summary of our techniques}

Our algorithm builds on the connection between GEP and ordinary differential equations (ODEs). Consider the linear ODEs $B\x'(t) = 2\pi iA\x(t)$ with initial condition $\x(0) = \sum_j \beta_j \ket{E_j}$. The solution equals $\x(t) = \sum_j \beta_j e^{2\pi i\lambda_j t}\ket{E_j}$. 
To solve an ODE in a quantum computer, we usually use the discretization method to reduce it to a linear system such that quantum linear solvers are applicable \cite{berry2017quantum,berry2014high,childs2020quantum}.
Suppose we discretize the time interval $[0,\tau]$ into $p$ sub-intervals via $t_0=0, t_1 = h, \ldots, t_p = ph = \T$, then it turns out that the quantum state of the solution of the linear system is proportional to $\sum_{l=1}^p \|\x(t_l)\|\, \ket{l} \ket{\x(t_l)} = \sum_j \beta_j \sum_{l=1}^p e^{2\pi i\lambda_j t} \ket{l}  \ket{E_j}$. When $\lambda_j$ are all real, we can estimate them by applying inverse quantum Fourier transform to $\ket{l}$. This gives us the target state stated in Problem \ref{prob:problem 1}. 

The basic idea of our algorithm is straightforward. The technical part is the complexity analysis. Since it relates to the solving of a linear system in a quantum computer, we need to estimate the condition number of this linear system. This is indeed the most technical part of this whole paper. Below, we briefly introduce our basic idea about this estimation. Denote the coefficient matrix of the linear system as $M$, then $\|M\|$ is easy to bound. As for $\|M^{-1}\|$, by definition it equals $\max_{\|\b\|=1} \|M^{-1}\b\|$. Since $M$ is obtained from solving an ODE and $M^{-1}\b$ can be viewed as the solution of the linear system $M\y = \b$, we can view $M^{-1}\b$ as an approximate solution of a linear ODE of the form $B\x'(t) = 2\pi iA\x(t) + \c(t)$ for some $\c(t)$ depending on $\b$. For this ODE, we know the explicit formula of its solution, from which we can bound $\|M^{-1}\b\|$ through some matrix inequalities.

\subsection{Organization of the paper}

This paper is organized as follows:
In Section \ref{section:Preliminaries}, we introduce some preliminary results that will be used in this paper. This includes some classical results about GEPs and the framework of block-encoding.
In Section \ref{section:Quantum differential equation method}, we present our quantum algorithm in detail. In Section \ref{section:Block-encoding method}, for comparison, we analyze the complexity of the standard algorithm based on QPE. Finally, in Section \ref{section:ow bounds for the generalized eigenvalue problem}, we give a simple lower bound analysis about solving singular GEPs in a quantum computer.

{\bf Notation.} For a matrix $A=(a_{ij})$, we use $A^{-1}$ to denote the Moore-Penrose inverse, $\|A\|$ to denote the operator norm (i.e., the maximal singular value), $\|A\|_{\max}=\max_{i,j}|a_{ij}|$ and $\|A\|_F = \sqrt{\sum_{i,j}|a_{ij}|^2}$. By $\kappa_A$, we mean the ratio of maximal singular value and minimal nonzero singular value, and still call it the condition number of $A$. We say $A$ is $s$-sparse if the maximal number of nonzero entries in each row and column is $s$. For two Hermitian matrices $A,B$, by $A\preceq B$ we mean $B-A$ is positive definite.

\section{Preliminaries}
\label{section:Preliminaries}

\subsection{Generalized eigenvalue problem}

Let $A,B$ be two $n$-by-$n$ complex matrices, the generalized eigenvalue problem (GEP) aims to find
all $\lambda \in \mathbb{C}$ and all vectors 
$\x\in \mathbb{C}^n$ such that
\be \label{GEP}
A\x = \lambda B \x.
\ee

Although the GEP looks like a simple generalization
of the standard eigenvalue problem, it exhibits some significant differences.
For example,
\begin{itemize}
\item It is possible for $\det(A-\lambda B)=0$ independent of $\lambda$ so that every scale is an eigenvalue, e.g., $A,B$ are singular such that $A\x=B\x=0$ for some nonzero $\x$. When this happens, we call $(A,B)$ a non-regular pair. Obviously, if $(A,B)$ is non-regular, then $B$ is singular.
\item If $B$ is singular, then $\infty$ can be an
eigenvalue. To see this, consider $B\x=\lambda^{-1}A\x$. Choose $\x$ such that $B\x=0, A\x\neq 0$, then $\lambda=\infty$. Because of this,
sometimes we write the GEP in the cross-product form
$\mu A \x = \lambda B \x$, and view the pair $(\lambda,\mu)$ as an eigenvalue. In this notation,
the eigenvalue $\infty$ corresponds to $(\lambda,0)$. When $(A,B)$ is regular, this matrix pair always has $n$ eigenvalues, including $\infty$. 
\end{itemize}

The above two unusual properties appear when $B$ is singular. Also, when $B$ is singular, the GEP is well-known to be ill-conditioned \cite{hochstenbach2019solving,van1979computation}. So it is hard to solve even for a quantum computer, see the lower bound analysis in Section \ref{section:ow bounds for the generalized eigenvalue problem}. Due to the above reasons, in this paper, we will mainly focus on the case  that $B$ is invertible. 

Usually, $\lambda$ may not be real even if $A,B$ are Hermitians.
The following simple result is a sufficient condition, which holds for many machine learning models \cite{de2005eigenproblems}.\footnote{A more general result is given in \cite[Theorem 8.7.1]{Golub}.}

\begin{prop}[see \cite{Golub}]
\label{lemma:real generalized eigenvalues}
Assume that $A,B$ are Hermitian matrices with $B$ positive  definite. Then 
there is an invertible matrix $E$ and a real diagonal matrix $\Lambda$ such that
$AE=BE\Lambda$. Moreover, the condition number of $E$ equals
$\kappa_E = \sqrt{\kappa_B}$.
\end{prop}

\begin{proof}
Since $B$ is Hermitian, it can be diagonalized by a
unitary matrix. Denote this decomposition as $U^\dag B U = D$, where $U$ is unitary and $D$ is diagonal.
Denote $C = B^{-1/2} = U D^{-1/2}$, then $C^\dag B C = I$.
Consider $\widetilde{A}:=C^\dag A C$. It is Hermitian, so there is a unitary $V$ and a diagonal matrix $\Lambda$ such that $V^\dag \widetilde{A} V = \Lambda$. Set $E=CV$, then
$E^\dag A E = \Lambda$ and $E^\dag B E = I$.
This implies that $A E = B E \Lambda$.
Therefore, the real matrix $\Lambda$ is the diagonal matrix
of the generalized eigenvalues.
As for the condition number, since $E=\sqrt{B^{-1}} V$ with $V$ unitary, we have $\kappa_E = \sqrt{\kappa_B}$.
\end{proof}

Even though the assumption that $B$ is
positive definite occurs in many cases in practice, we can make an even weaker assumption called definite pairs such that all the eigenvalues are real. A Hermitian pair
$(A,B)$ is called definite if 
\be
\label{def:gamma}
\gamma(A,B) := \min_{\x \in \mathbb{C}^n,\|\x\|_2=1} 
\sqrt{(\x^\dagger A \x)^2
+(\x^\dagger B \x)^2} > 0.
\ee
The quantity $\gamma(A,B)$ is known as the Crawford number of the pair $(A,B)$.
The basic fact about definite pairs is that they can be transformed  by a unitary matrix into a Hermitian pair $(A', B')$ in which $B'$ is positive definite.
More precisely, 

\begin{prop}
[Theorem 1.18 of Chapter VI in \cite{stewart1990matrix}]
\label{lem:definite pair}
Let $(A,B)$ be a definite pair. For any $\theta\in[0,2\pi)$ let $R(\theta)$ be the 2-dimensional rotation of angle $\theta$ and 
\be
\left[ \begin{array}{c}
A_\theta \\
B_\theta
\end{array} \right] :=
R(\theta)\otimes I_n
\left[ \begin{array}{c}
A \\
B
\end{array} \right]
=\left[ \begin{array}{c}
A\cos\theta - B \sin\theta \\
A\sin\theta + B \cos\theta
\end{array} \right].
\ee
Then there is a $\theta\in[0,2\pi)$
such that $B_\theta$ is positive definite
and $\gamma(A,B)$ is the smallest singular value of $B_\theta$.
\end{prop}

As a result of Propositions \ref{lemma:real generalized eigenvalues} and \ref{lem:definite pair}, if
$(A,B)$ is a definite pair, then there is a nonsingular matrix $E$ such that $E^\dag A E$ and $E^\dag B E$ are diagonal. Moreover, the eigenvalues are all real.

Next, we briefly introduce the conditioning analysis about computing eigenvalues. It turns out that the complexity of the quantum algorithm proposed in this paper will depend on the quantities (i.e., $\kappa_E,\kappa_B,\|A^{-1}\|$) that describe the conditioning, which seems not unrealistic.

The conditioning describes the stability of eigenvalues after small perturbations. According to Stewart \cite{stewart1978perturbation}, the appropriate measure for generalized eigenvalue perturbations is the chordal metric 
${\tt chord}(a,b)$ defined by
\be
{\tt chord}(a,b) := \frac{|a-b|}{\sqrt{1+|a|^2} \sqrt{1+|b|^2}}, \quad \text{where}~a,b\in \mathbb{C}.
\ee
Stewart showed that if $\lambda$ is a distinct eigenvalue of $(A,B)$ and $\lambda'$ is the corresponding eigenvalue of the perturbed pair $(A', B')$ with $\|A-A'\|\approx \|B-B'\| \approx \epsilon$, then
\be
{\tt chord}(\lambda, \lambda') \leq \frac{\epsilon}{\sqrt{(\y^\dag A \x)^2 + (\y^\dag B \x)^2}} + O(\epsilon^2),
\ee
where $\x,\y$ are unit and satisfy $A\x = \lambda B\x, \y^\dag A = \lambda \y^\dag B$. So $1/\sqrt{(\y^\dag A \x)^2 + (\y^\dag B \x)^2}$ serves as a condition number for the eigenvalue $\lambda$.
Indeed, $\gamma^{-1}(A,B)$ can be viewed as an
upper bound of the condition number of all eigenvalues \cite[Theorem 3.3 of Chapter VI]{stewart1990matrix}. Note that
$\gamma(A,B) \geq \sqrt{\|A^{-1}\|^{-2}+\|B^{-1}\|^{-2}}$.
So 
\be
\label{upper bound of gamma}
\gamma^{-1}(A,B)\leq \sqrt{
\frac{\|A^{-1}\|^{2}\|B^{-1}\|^{2}}{\|A^{-1}\|^{2}+\|B^{-1}\|^{2}}}
\leq 
\sqrt{\|A^{-1}\|\|B^{-1}\|/2}.
\ee

Generally,  let $E$ be the matrix of the
generalized eigenvectors with condition number $\kappa_E$, then by \cite[Theorem 2.6 of Chapter VI]{stewart1990matrix} the condition number
is bounded by $\kappa_E \rho_L(A,B)$, where 
$\rho_L$ is a metric that depends on the perturbations on $A,B$. 

The conditioning of GEP also 
relates to the condition
number of $B$, for instance see 
\cite[Theorem 2]{crawford1976stable}. 
Consider the GEP $A\x=\lambda B\x$ with $A, B$ Hermitian and $B$ is positive definite.
We use $A' = A + \Delta A, B' = B + \Delta B$ to denote
the perturbations of $A,B$ with $\|\Delta A\|, \|\Delta B\|\leq \epsilon$. 
Then for the $i$-th eigenvalue $\lambda_i$ of $(A,B)$ and $i$-th eigenvalue ${\lambda}'_i$ of $(A', B')$, we have
$
|\lambda_i' - \lambda_i| \leq  
(1 + |\lambda_i|) \epsilon \kappa_B.
$

\subsection{Block-encoding}
\label{subsec:Block-encoding}

In this section, we briefly introduce some necessary results about block-encoding. The main references are \cite{chakraborty2018power,Gilyen-QSVT}.

\begin{defn}[Block-encoding]
Suppose that $A$ is a $p$-qubit operator, $\alpha,\epsilon\in \mathbb{R}^+$ and $q\in\mathbb{N}$. Then we say that the $(p+q)$-qubit unitary $U$ is an $(\alpha,q,\epsilon)$-block-encoding of $A$, if
\be
\|A-\alpha (\langle 0|^{\otimes q}\otimes I ) U (|0\rangle^{\otimes q}\otimes I)\|
\leq \epsilon.
\ee
\end{defn}

In matrix form, we can view
$U =
\left[ \begin{array}{ccccccc}
A/\alpha & \cdot \\
\cdot & \cdot \\
\end{array} \right].$ So to verify if a unitary $U$ is a block-encoding of $A$, we only need to verify if $U\ket{0}\ket{x} = \ket{0}(A/\alpha)\ket{x} + \ket{1}\ket{x'}$ for all $\ket{x}$, where $\ket{x'}$ is a garbage state. The following lemma is about the construction of block-encodings of sparse matrices.

\begin{lem}[Lemma 48 of \cite{Gilyen-QSVT}]
\label{sparse block-encoding}
Assume that $A$ is an $s$-sparse $n$-by-$n$ matrix given in the sparse-access input model. Then for any $\epsilon\in (0,1)$, we can implement an $(s\|A\|_{\max},{\rm poly}\log(n/\epsilon),$ $\|A\|_{\max}\epsilon)$-block-encoding of $A$ in time $O({\rm poly}\log(n/\epsilon))$.
\end{lem}

In the original statement of Lemma 48 of \cite{Gilyen-QSVT}, the authors assume that $\|A\|_{\max}\leq 1$.
So to apply their result, we can consider $A/\|A\|_{\max}$.
The above lemma then comes from the following simple fact.

\begin{fact}{\rm
\label{fact of BE}
An $(\alpha,q,\epsilon)$-block-encoding of $A/\beta$ is an
$(\alpha \beta,q,\beta\epsilon)$-block-encoding of $A$. 
}\end{fact}

\begin{lem}[Lemmas 9 and 10 of \cite{chakraborty2018power}]
\label{negative power of block-encoding}
Assume that $\delta\in(0,1), \kappa\geq 2$.
Suppose that $B$ is Hermitian with $I/\kappa \preceq B \preceq I$,
and we have an $(\alpha,q,\epsilon)$-block-encoding of $B$ that can be implemented using $O(T)$ elementary gates.

\begin{itemize}
\item If
$\epsilon = o(\delta\kappa^{-1.5}\log^{-3}(\kappa^{1.5}/\delta))$, then we can implement a $(2\sqrt{\kappa}, q + O(\log(\kappa^{1.5} \log 1/\delta), \delta)$-block-encoding of $B^{-1/2}$ in cost
$
O(\alpha \kappa (q+T) \log^{2}(\kappa^{1.5}/\delta)).
$
\item If $\epsilon = o(\delta\kappa^{-1}\log^{-3}(\kappa /\delta))$,
then we can  implement a $(2 , q + O(\log \log 1/\delta), \delta)$-block-encoding of $B^{1/2}$ in cost
$
O(\alpha \kappa (q+T) \log^{2}(\kappa/\delta)).
$
\end{itemize}

\end{lem}

\begin{lem}[Lemma 4 of \cite{chakraborty2018power}]
\label{lem:product block-encoding}
Suppose  for $i=1,2$ we have an $(\alpha_i,q_i,\epsilon_i)$-block-encoding of $A_i$
that is constructed in time $O(T_i)$. Then we can create
an $(\alpha_1\alpha_2,q_1+q_2,\alpha_1 \epsilon_2+\alpha_2\epsilon_1)$-block-encoding of $A_1A_2$ in time $O(T_1+T_2)$.
\end{lem}

\begin{defn}[Quantum phase estimation] 
\label{sve problem} 
Let $A$ be an $n$-by-$n$ Hermitian matrix with eigenvalue decomposition $A= \sum_{k=1}^n \lambda_{k} |u_{k}\rangle \langle u_{k}|$. Let $\epsilon\in(0,1)$. The quantum phase estimation (QPE) problem with accuracy $\epsilon$ is defined as: 
Given access to
$\sum_{k=1}^n \beta_k \ket{u_k}$ 
to perform the mapping  
\be
 \sum_{k=1}^n  \beta_k  \ket{0} \ket{u_k}
 \mapsto \sum_{k=1}^n  \beta_k \ket{\tilde{\lambda}_k} \ket{u_k},
\ee
such that $|\tilde{\lambda}_k - \lambda_k| \leq \epsilon$ for all $k \in \{1,2,\ldots,n\}$.
\end{defn}

\begin{lem}[Theorem 27 of \cite{chakraborty2018power}]
\label{quantum SVE}
Let $\epsilon, \tilde{\epsilon}\in (0, 1)$, and $\epsilon' = \tilde{\epsilon} \epsilon /( 4 \log^2(1/\epsilon))$. 
Given an $(\alpha,q,\epsilon')$-block-encoding Hermitian matrix $A$ that is implemented in cost $O(T)$, then there is a quantum algorithm that solves the QPE problem of $A$ with accuracy $\epsilon$, with success probability at least $1-\tilde{\epsilon}$ in time
$
O( (T_{\rm in}+{\alpha}\epsilon^{-1} (q+T) ) {\rm poly}\log(1/\tilde{\epsilon})),
$
where $O(T_{\rm in})$ is the cost to prepare the 
initial state.
\end{lem}

In \cite{chakraborty2018power}, Theorem 27 is used to estimate the singular values. However, it is not hard to modify their algorithm to estimate the eigenvalues of Hermitian matrices.
The following lemma gives the complexity of solving linear systems in the framework of block-encoding.

\begin{lem}[Theorem 30 of \cite{chakraborty2018power}]
\label{lem:quantum linear solver}
Let $\kappa\geq 2$, and $H$ be a 
matrix with non-zero singular values lie in the range $[-1,-1/\kappa]\cup[1/\kappa, 1].$
Suppose that for $\delta = o(\epsilon/(\kappa^2 \log^3 (\kappa/\epsilon))$ we have an $(\alpha, q,\delta)$-block-encoding of $H$ that can be implemented using $T_U$ elementary gates. Also suppose that we can prepare an input state $\ket{\psi}$ which spans the eigenvectors of $H$ in time $O(T_{\rm in})$. Then there is a quantum algorithm that outputs a state that is $\epsilon$-close to $H^{-1}\ket{\psi}/\|H^{-1}\ket{\psi}\|$ at a cost
\be \label{lem-complexity:quantum linear solver}
\widetilde{O}(\kappa (T_{\rm in}+\alpha(q+T_U) \log^2(\kappa/\epsilon)) \log \kappa).
\ee
\end{lem}

In the above lemma, 
the accuracy $\delta$ in the block-encoding of $H$ should be much
smaller than $\epsilon$, 
the accuracy to approximate the 
target state. 
Usually this will not cause any
problem to the
complexity analysis
because $\delta$ appears as a logarithmic term.
Another point of Lemma \ref{lem:quantum linear solver}
is the assumption $\|H\|\leq 1$. For solving a linear
system $H \x = \b$, we can consider $(H/\|H\|) \x = \b/\|H\|$ instead.
By Fact \ref{fact of BE}, if we only have an $(\alpha,q,\delta)$-block-encoding of $H$, then
$\alpha$ in the complexity (\ref{lem-complexity:quantum linear solver}) should be divided by $\|H\|$.

\section{Quantum differential method for GEPs}
\label{section:Quantum differential equation method}

In this section,  based on the connection between GEPs and ODEs, we propose a quantum algorithm to solve Problem \ref{prob:problem 1}. The main idea is pretty straightforward as follows.
Consider the following linear system of ODEs with initial state $\x(0)$
\be \label{simple ODE}
B\frac{\d\x(t)}{\d t} = 2\pi i A \x(t).
\ee
Denote the eigenpairs of GEP $A\x=\lambda B\x$ 
as $\{(\lambda_j,|E_j\rangle): j=1,\ldots,n\}.$
If $B$ is invertible, then the solution is
$\x(t) = e^{2\pi iB^{-1}A t} \x(0)$. Thus if $\x(0)$ is a 
linear combination of the eigenvectors, say
$\x(0) = \sum_{j=1}^n \beta_j |E_j\rangle$, then
$\x(t) = \sum_{j=1}^n \beta_j e^{2\pi i\lambda_j t} |E_j\rangle$. 

To solve an ODE numerically, a useful method is discretization. If we 
discretize the time interval $[0,\T]$ via
$t_0=0,t_1=h, t_2=2h, \ldots, t_{p}=ph=\T$, the ODE (\ref{simple ODE}) reduces to a linear system whose solution is $(\x(t_1)^T,\ldots,\x(t_p)^T)^T$, where $T$ refers to the transpose operation. In a quantum computer, when a quantum linear solver is used, we will obtain the quantum state of this solution, which is proportional to
\be
\label{superposition solution}
\tilde{\x}  = \frac{1}{\sqrt{p}}
\sum_{l=1}^{p}\|\x(t_l)\|~|l\rangle |\x(t_l)\rangle = \frac{1}{\sqrt{p}}
\sum_{j=1}^n \beta_j \sum_{l=1}^{p}e^{2\pi i\lambda_j lh}|l\rangle |E_j\rangle.
\ee
Applying the inverse quantum Fourier transform to $|l\rangle$, we obtain the following expected state, which contains the information of eigenpairs, with high probability
\be \label{state after QFT}
\sum_{j=1}^n \beta_j |\tilde{\lambda}_j\rangle |E_j\rangle,
\ee
where $\tilde{\lambda}_j$ is an approximation of $\lambda_j$ up to 
additive error $1/(ph)$.
The error analysis here is
similar to that in quantum phase estimation, so we defer it to Appendix \ref{appendix:Error analysis}. 
For simplicity, we shall use $\rho(A,B)$
or just $\rho$ to denote an upper bound of the eigenvalues, i.e., $\rho\geq \max\{1,\max_j |\lambda_j|\}$.
As for $h$, we should make sure that
$|\lambda_jh| \leq 1$ for all $j$,
so $h = O(1/\rho)$.
As a result, we can set
$p = O(\rho/\epsilon) $.
In conclusion, the parameters we choose are
\be
\label{important parameters}
p = O(\rho/\epsilon), \quad
h = O(1/\rho), \quad
\T = O(1/\epsilon).
\ee

In (\ref{important parameters}), when determining the parameters, we only considered the error analysis from a similar idea to QPE. Later, we need to modify them by considering the error caused by the discretization method. The modification is slight, so (\ref{important parameters}) is almost our final choice.

\begin{rem}{\rm
A main ingredient of QPE is the Hamiltonian simulation. It is used in QPE to generate a state proportional to
$\sum_{j,l} \beta_j e^{2\pi i \lambda_j l/p} \ket{l}\ket{E_j}$. This state is similar to the one we obtained in (\ref{superposition solution}). Hamiltonian simulation only works for Hermitian matrices, so we cannot apply it to our problem directly. However, the original goal of Hamiltonian simulation is to simulate quantum systems, i.e., to solve the Schr\"{o}dinger equation. With the discovery of quantum linear solvers, we know how to solve differential equations in a quantum computer. So our main idea can be simply described as replacing Hamiltonian simulation by solving differential equations in QPE.
}\end{rem}

In a quantum computer, there are several general-purpose quantum algorithms for solving differential equations
\cite{berry2017quantum,berry2014high,childs2020quantum,childs2020high}.
The basic idea of these algorithms is to reduce the linear differential equations into a linear system of equations through some discretization methods, then use a quantum linear solver to obtain the solution state.
For the differential equation
(\ref{simple ODE}), the quantum algorithm proposed in
\cite{berry2017quantum} may not be suitable for our problem since this algorithm tries to approximate
$e^{iB^{-1}At}$ by Taylor expansion.
Even if it is easy to compute $B^{-1}$ in a quantum computer, the linear system they constructed depends on the matrix $B^{-1}A$.
So we still need to know the entries of $B^{-1}A$.
The quantum spectral method \cite{childs2020quantum}
and the higher order method \cite{berry2014high}
can be used to solve the differential equation (\ref{simple ODE}), while they are not appropriate for our purpose.
The reason is that the higher order method \cite{berry2014high} has worse dependence on the precision. The spectral method \cite{childs2020quantum} is based on Chebyshev polynomials.
The interpolation nodes they used are the Chebyshev-Gauss-Lobatto quadrature nodes, i.e.,
$t_l = \cos(l\pi/n)$, here we have to use the nodes $t_l=l(h)$. 
In the following, we will apply the spectral method based on the Fourier series to solve the differential system (\ref{simple ODE}). This method has been used in \cite{childs2020high} to solve some PDEs. Since the Fourier spectral method has not been used previously to solve ODEs in a quantum computer, a technical part below is the estimation of its complexity (including error analysis and estimation of condition number).

\subsection{Fourier approximation}

In this section, we aim to give an introduction to the Fourier spectral method to solve the linear system of ODEs (\ref{simple ODE}).
For more about this method, we refer to the book \cite{shen2011spectral}.
For any function $f(t)$, where $t\in[0,L]$, the $K$-truncated Fourier series of $f(t)$ is defined by
\be \label{Fourier series}
f_K(t) = \sum_{k=-K}^K c_k e^{2\pi i kt/L} = \sum_{k=0}^{2K} c_{k-K} e^{2\pi i (k-K)t/L} ,
\ee
where
\[
c_k = \frac{1}{L} \int_0^L f(t) e^{-2\pi i kt/L} \d t.
\]

\begin{prop} \cite[Theorem 2.3]{adcock2014resolution}
\label{lem:error of Fourier series}
Suppose $f$ is analytic, then there exist $R_f=O(\max|f(x)|)$ and $S=\cot^2(\pi/4L)$ such that
$
|f(t)-f_K(t)| \leq R_f S^{-K}
$ for all $t\in[0,L]$.
\end{prop}

\begin{fact}{\rm
\label{fact:norm bound}
Given $\a=(a_1,\ldots,a_n), \b=(b_1,\ldots,b_n)$
such that $|a_i-b_i|\leq R_fS^{-K}$ for all $i$,
then $\|\a-\b\|\leq \sqrt{n} R_fS^{-K}$. For us, $L = \T = O(1/\epsilon)$, so $S \approx 1/\epsilon^2$.
So to make sure the error is bounded by
$\epsilon$, it suffices to choose $K=O((\log \epsilon)^{-1}\log (R_fn/\epsilon))$. 
}\end{fact}

Combining (\ref{important parameters}), we now choose $p=O(\max\{\rho/\epsilon, (\log \epsilon)^{-1}\log (R_{\x}n/\epsilon) \} ) = \widetilde{O}(\rho/\epsilon)$ odd.\footnote{In the complexity analysis of our quantum algorithms below, we will ignore the logarithmic terms for simplicity, so we may think here that $p = O(\rho/\epsilon)$, i.e., equation (\ref{important parameters}). Choosing odd $p$ is to ensure that $(p-1)/2$ is an integer. The constant $R_{\x}$ is upper bounded $\|\x(t)\|=O(\kappa_E)$.}
For any $l\in\{0,1,\ldots,p-1\}$, from (\ref{Fourier series})
the $j$-th entry of $\x$ and its derivative at time $lh$ can be approximated by
\bea \label{fourier expanding}
\hat{x}_j(lh) &:=& 
\frac{1}{\sqrt{p}}\sum_{k=0}^{p-1} c_{jk} e^{2\pi i (k-\frac{p-1}{2}) \frac{lh}{\T}}
=\frac{1}{\sqrt{p}}\sum_{k=0}^{p-1} c_{jk} e^{2\pi i (k-\frac{p-1}{2}) \frac{l}{p}}
=
\frac{e^{\pi i \frac{(1-p)l}{p}}}{\sqrt{p}}
\sum_{k=0}^{p-1} c_{jk} e^{2\pi i \frac{kl}{p}},  \label{Fourier expanding} \\
\hat{x}_j'(lh) &=& \frac{2\pi i}{\sqrt{p}\T} \sum_{k=0}^{p-1} (k-\frac{p-1}{2}) c_{jk} e^{2\pi i (k-\frac{p-1}{2}) \frac{lh}{\T} } = \frac{2\pi ie^{\pi i \frac{(1-p)l}{p}}}{\sqrt{p}\T} \sum_{k=0}^{p-1} (k-\frac{p-1}{2}) c_{jk} e^{2\pi i \frac{kl}{p} },
 \label{fourier expanding1}
\eea
where $c_{jk}$ are the unknown Fourier coefficients we aim to compute.
In the above, we used the fact that $\T=ph$.
Here, we do not consider the case $l=p$, i.e., $\x(T)$. This makes the right hand sides of (\ref{fourier expanding}), (\ref{fourier expanding1})
the quantum Fourier transform.
The constant $1/\sqrt{p}$ is added for normalization.

In matrix form, denote 
\bea
C &=& (c_{jk})_{1\leq j\leq n, 0\leq k \leq p-1}, \label{matrix-notation1} \\
D &=& {\rm diag}\{k-\frac{p-1}{2}:0\leq k\leq p-1\}, \label{matrix-notation2} \\
F &=&  \frac{1}{\sqrt{p}}(e^{2\pi i kl/p})_{0\leq k,l \leq p-1}. \label{matrix-notation3}
\eea 
By (\ref{fourier expanding}), (\ref{fourier expanding1}), when $l\geq 1$, we have \footnote{Here we use $\hat{\x}(t)$ to denote the approximation of $\x(t)$, the exact solution of the ODE (\ref{simple ODE}).}
\be
\label{approximation in matrix form}
\hat{\x}(lh) = e^{\pi i \frac{(1-p)l}{p}} C F |l\rangle,
\quad
\hat{\x}'(lh) = \frac{2\pi ie^{\pi i \frac{(1-p)l}{p}}}{\T} C D F |l\rangle.
\ee
When $l=0$, 
$
\x(0) = CF|0\rangle
$
is the initial condition.
From (\ref{simple ODE}),
we know that
\bea
\label{initial linear system}
\begin{cases}
\x(0) =   CF|0\rangle, \\
\displaystyle \frac{1}{\T} B C D F|l\rangle =  ACF|l\rangle, \quad 1\leq l \leq p-1.
\end{cases}
\eea

\begin{defn}[Vectorization of matrices]
Let $A=(a_{ij})_{m\times n}$ be a matrix, its vectorization is an $mn$-dimensional column vector
$
\VEC(A) = 
(a_{11},\ldots,a_{m1}, \ldots, a_{1n},\ldots,a_{mn})^T.
$
So for any three matrices $A,B,C$ we have $\VEC(ABC)=(C^T\otimes A) \VEC(B)$.
\end{defn}

Now let $\VEC(C) =(c_{10},\ldots,c_{n0},\ldots,c_{1(p-1)},\ldots,c_{n(p-1)})^T$ denote the vectorization of $C$, then we obtain a linear system of $\VEC(C)$ through (\ref{initial linear system}):
\be\label{linear system}
\left[ \begin{array}{ccccccc}
\langle 0| F^T \otimes I \\
\langle 1| F^T \otimes A - \T^{-1} \langle 1| F^T D \otimes B \\
\cdots \cdots \cdots\\
\langle p-1| F^T \otimes A - \T^{-1} \langle p-1| F^T D \otimes B 
\end{array} \right] \VEC(C) =
\left[ \begin{array}{ccccccc}
\x(0) \\
0 \\
\cdots \\
0 
\end{array} \right].
\ee
The above coefficient matrix is $np\times np$
\be \label{coeff matrix}
M = \left[ \begin{array}{ccccccc}
\langle 0| F^T \otimes I \\
(F_0^T \otimes I) (I\otimes A-\T^{-1} D\otimes B)
\end{array} \right],
\ee
where $F_0$ is obtained by deleting the first column of $F$.

Before starting more details, we first discuss the relationship between
$\VEC(C)$ and the state (\ref{state after QFT}) we aim to prepare.
By definition, the quantum state of $\VEC(C)$ equals
\be \label{state of C}
\ket{\VEC(C)} = \frac{1}{\|C\|_F} \sum_{j=1}^n \sum_{k=0}^{p-1} c_{jk}|k,j\rangle,
\ee
where $\|C\|_F$ is the Frobenius norm.
By the first identity in (\ref{approximation in matrix form}) and note that $F$ is unitary, we have
\be
\|C\|_F  = \sqrt{\sum_{l=0}^{p-1} \|\hat{\x}(lh)\|^2 }.
\ee
It is easy to see that
the superposition of the approximate solutions (see equation (\ref{superposition solution})) is 
\be \label{state of the solution}
\frac{1}{\|C\|_F}\sum_{l=0}^{p-1}\|\hat{\x}(lh)\|~ |l\rangle |\hat{\x}(lh)\rangle
= \frac{1}{\|C\|_F\sqrt{p}} 
\sum_{l=0}^{p-1} e^{\pi i \frac{(1-p)l}{p}}
\sum_{j=1}^n 
\sum_{k=0}^{p-1} c_{jk} e^{2\pi i kl/p}|l,j\rangle.
\ee
Denote
\be \label{unitaryUp}
U_p = {\rm diag}\{e^{\pi i \frac{(1-p)l}{p}} : 
l=0,\ldots,p-1\}.
\ee
Then we can obtain the state (\ref{state of the solution}) by applying  $U_pF$ to $|k\rangle$ in $\ket{\VEC(C)}$. 
To estimate
the eigenvalues (i.e., to obtain (\ref{state after QFT})), we only need to apply $F^{-1}$ to $|l\rangle$ in the state (\ref{state of the solution}).
Since $F,U_p$ can be implemented efficiently in the quantum circuits, it follows that to determine the overall cost of solving Problem \ref{prob:problem 1} by the above idea, it suffices to know the cost of solving the linear system (\ref{linear system}) to prepare $|\VEC(C)\rangle$. 
For this, we first need to estimate the error of the solution (i.e., the error between (\ref{superposition solution}) and (\ref{state of the solution})) and the condition number of the linear system (\ref{linear system}).

\subsection{Analysis of the error and condition number }

\begin{prop}[Error analysis]
\label{prop:Error analysis}
Let $\tilde{\x}$ be the superposition (\ref{superposition solution}) of the exact solutions of the
differential equation (\ref{simple ODE}), let $|C\rangle$ be the quantum state of the solution of the linear system (\ref{linear system}) obtained by a quantum linear solver up to precision $\epsilon$, then
\be
\Big\| |\tilde{\x}\rangle - (U_pF\otimes I)|C\rangle \Big\| \leq  2\epsilon.
\ee
\end{prop}

\begin{proof}
Denote $|\VEC(C)\rangle$ as the quantum state of the exact solution 
of the linear system (\ref{linear system}),
then
$\| |\VEC(C)\rangle -  |C\rangle \|\leq \epsilon$ by assumption.
Before normalization, denote
\[
\hat{\x} 
= \frac{1}{\sqrt{p}}
\sum_{l=0}^{p-1}\|\hat{\x}(t_l)\| \, |l\rangle |\hat{\x}(t_l)\rangle
= \frac{\|C\|_F}{\sqrt{p}} (U_pF\otimes I)
|\VEC(C)\rangle.
\]
By Proposition \ref{lem:error of Fourier series}
and Fact \ref{fact:norm bound}, $\|\x(t_l)-\hat{\x}(t_l)\|\leq \sqrt{n}R_{\x}S^{-p}$. 
From equation (\ref{superposition solution}),
\[
\tilde{\x}  = \frac{1}{\sqrt{p}}
\sum_{l=0}^{p-1}\|\x(t_l)\| \, |l\rangle |\x(t_l)\rangle .
\]
So we have
$
\|\tilde{\x} - \hat{\x}\| \leq \sqrt{n} R_{\x} S^{-p}.
$
After normailzation, we have
$
\||\tilde{\x}\rangle - |\hat{\x}\rangle\| \leq 2\sqrt{n}R_{\x}S^{-p}/\|\tilde{\x}\|.
$
Notice that $\|\x(t_0)\| = 1$, so
\[
\|\tilde{\x}\|^2
= \frac{1}{p} \sum_{l=0}^{p-1}
\|\x(t_l)\|^2 
\geq \frac{1}{p}.
\]
Thus
$
\||\tilde{\x}\rangle - |\hat{\x}\rangle\| 
\leq 2\sqrt{pn}R_{\x}S^{-p}\ll \epsilon$ 
since $p=O(\max\{\rho/\epsilon, (\log \epsilon)^{-1}\log (R_{\x}n/\epsilon) \} )$ and $S\approx 1/\epsilon^2$.
Combining the error 
between $|\VEC(C)\rangle$ and $|C\rangle$, the fact that $U_pF$ is unitary and
$|\hat{\x}\rangle = (U_pF\otimes I)
|\VEC(C)\rangle$, we obtain the claimed result.
\end{proof}

The proof of the following result on the condition number is the most technical one of this paper. The basic intuition is as follows: the condition number of $M$ is defined by $\|M\| \|M^{-1}\|$. 
It is easy to bound $\|M\|$ using the triangular inequality. As for $\|M^{-1}\|$, by definition it equals $\max_{\b:\|\b\|=1} \|M^{-1}\b\|$. Since $M$ is constructed from the discretization of the ODE (\ref{simple ODE}), it follows that $M^{-1}\b$ is closely related to the solution of this ODE. For the ODE (\ref{simple ODE}), there is an explicit formula for the solution. Based on this formula, we can bound $\|M^{-1}\|$.

\begin{prop}[Condition number]
\label{prop:condition number}
Let $A,B$ be the matrices that satisfy the conditions described in Problem \ref{prob:problem 1}.
Let $E$ be the matrix of the generalized eigenvectors with condition number $\kappa_E$, 
and $M$ be the matrix (\ref{coeff matrix}).
Then 
\be
\|M\|  =  \Theta(\|I\otimes A-\T^{-1} D\otimes B\|),
\ee
and
\be
\|M^{-1}\| \leq O(\kappa_E  \|B^{-1}\|/\epsilon).
\ee
Consequently, the condition number of $M$  is upper bounded by 
\be
O\Big(
\kappa_E
(\|A\|+\rho \|B\| ) \|B^{-1}\|/\epsilon
\Big).
\ee
\end{prop}

\begin{proof}

Set $|e\rangle=\frac{1}{\sqrt{p}}(1,\ldots,1)^T \in \mathbb{R}^p$, and
$
I_0 = [|1\rangle, \ldots, |p-1\rangle]_{p\times (p-1)}
$, the $p\times p$ identity matrix with
the first column removed,
then $F_0=FI_0, F|0\rangle =  |e\rangle$. As for the condition number of $M$, we consider
\beas
M^\dag M &=& |e\rangle \langle e|\otimes I + (I\otimes A-\T^{-1}D\otimes B)^\dag (\overline{F}I_0 I_0^T F^T \otimes I)(I\otimes A-\T^{-1}D\otimes B) \\
&=& |e\rangle \langle e|\otimes I + (I\otimes A-\T^{-1}D\otimes B)^\dag (\overline{F} (I- |0\rangle\langle 0|) F^T \otimes I)(I\otimes A-\T^{-1}D\otimes B) \\
&=& |e\rangle \langle e|\otimes I +  (I\otimes A-\T^{-1}D\otimes B)^\dag 
(( I - |e\rangle \langle e|) \otimes I)(I\otimes A-\T^{-1}D\otimes B) .
\eeas
Since $\||e\rangle \langle e|\otimes I \| =\|(I-|e\rangle \langle e|) \otimes I \|=1$, we have
\be \label{cond:bound 1}
\|M\|^2 = \|M^\dag M\|\leq  1 +  \|I\otimes A-\T^{-1}D\otimes B\|^2.
\ee

On the other hand, since $M':=(I_0^T F \otimes I) (I\otimes A-\T^{-1} D\otimes B)$ is a submatrix of
$M$, we have $\|M\| \geq \|M'\|$.
For simplicity, denote $\T^{-1} D = -{\rm diag}(d_1,\ldots,d_p)$.
Also denote $N = I\otimes A -\T^{-1} D \otimes B$. Since $D$ is diagonal, we have
$\|N\| = \max_{1\leq j \leq p} \|A + d_j B\|$.
Without loss of generality, suppose
that $\|N\| = \|A + d_1 B\|$.
Denote the unit right and left singular vector
of $A + d_1 B$ as $|\x\rangle,|\y\rangle$, that is $(A + d_1 B)|\x\rangle = \|A + d_1 B\| \, |\y\rangle$. Then the corresponding 
unit right and left singular vector of $N$ are
$|0\rangle|\x\rangle, |0\rangle|\y\rangle$ respectively.
Thus
\be
\label{lower bound of spectral norm}
\|M\|\geq\|M'\| \geq \|M'|0\rangle|\x\rangle\|
= \|N\| \, \|(I_0^T F \otimes I)|0\rangle|\y\rangle\| 
= \|N\| \, \|(I_0^T \ket{e} \otimes \ket{\y})\| 
= \sqrt{\frac{p-1}{p}} \|N\|.
\ee
Combining (\ref{cond:bound 1}), we have
\be
\|M\| = \Theta(\|I\otimes A-\T^{-1} D\otimes B\|).
\ee


As for $\|M^{-1}\|$, it equals
\bes
\|M^{-1}\| = \quad
\max_{\b\in \mathbb{C}^{np} :~M \y = \b, ~ \|\b\|=1} 
\quad \|\y\|.
\ees
For any unit vector $\b^T=(\b(0)^T,\b(1)^T,\ldots,\b(p-1)^T) \in \mathbb{C}^{np}$, where $\b(i)\in \mathbb{C}^n$ is a column vector, from
(\ref{linear system}) we know that
\bes
\left[ \begin{array}{ccccccc}
\langle 0| F^T \otimes I \\
\langle 1| F^T \otimes A - \T^{-1}\langle 1| F^T D \otimes B \\
\cdots \cdots \cdots \\
\langle p-1| F^T \otimes A - \T^{-1}\langle p-1| F^T D \otimes B 
\end{array} \right] \y =  
\left[ \begin{array}{ccccccc}
\b(0) \\
\b(1) \\
\cdots \\
\b(p-1)
\end{array} \right].
\ees
Denote $Y$ as the $n\times p$ matrix such that its vectorization is $\y$, that is $\VEC(Y)=\y$.
Then $\|\y\| = \|Y\|_F$. We also have
\beas
\begin{cases}
YF|0\rangle =  \b(0),  \\
AYF|l\rangle - \T^{-1} BYDF|l\rangle =   \b(l) , \quad 1\leq l \leq p-1.
\end{cases}
\eeas

From the discretization procedure,
we can view $YF|l\rangle$ as an approximation of the solution of the following linear system of ODEs
\bea \label{ODE:condition number}
\begin{cases}
\displaystyle B\frac{\d \x(t)}{\d t} = 2\pi i A \x(t) + 2\pi i \c(t), \\
\x(0) =\b(0) .
\end{cases}
\eea
where $\c(t) = \b(l+1)$ if $lh < t\leq (l+1)h$. 

The solution of (\ref{ODE:condition number})
equals
\be
\x(t) = e^{2\pi i B^{-1}A t} \b(0) + 2\pi i e^{2\pi i B^{-1}A t}
\int_0^t e^{-2\pi i B^{-1}A s}
B^{-1} \c(s) ds.
\ee
In the following analysis, we will suppose $A$ is invertible for simplicity. When $A$ is not invertible, the estimations below are also correct. We will give the detailed analysis of this case in Appendix \ref{appendix:Supplementary proof}.
Assume that $t\in(0,h]$
and $l\in\{1,\ldots,p-1\}$, then
\bea
\x(t+lh) &=& e^{2\pi i B^{-1}A (t+lh)} \b(0) + 2\pi i e^{2\pi i B^{-1}A (t+lh)}
\sum_{j=0}^{l-1}
\int_{jh}^{(j+1)h}
e^{-2\pi i B^{-1}A s}
B^{-1} \b(j+1) ds \nonumber \\
&& +\, 2\pi i e^{2\pi i B^{-1}A (t+lh)}\int_{lh}^{t+lh}
e^{-2\pi i B^{-1}A s}
B^{-1} \b(l+1) ds  \nonumber \\
&=& e^{2\pi i B^{-1}A (t+lh)} \b(0) 
-
\sum_{j=0}^{l-1}
e^{2\pi i B^{-1}A (t+(l-j)h)}
(e^{-2\pi i B^{-1}A h} 
- I )(B^{-1}A)^{-1}
B^{-1} \b(j+1)  \nonumber \\
&& 
-\, e^{2\pi i B^{-1}A t}
(e^{-2\pi i B^{-1}A t} 
-
I )
(B^{-1}A)^{-1}
B^{-1} \b(l+1).
\label{solution of ODE:cond}
\eea

Notice that $B^{-1}A = E \Lambda E^{-1}$ is diagonalizable and $t\in(0,h]$ in the solution (\ref{solution of ODE:cond}), we have
\beas
&& \|e^{2\pi i B^{-1}A (t+(l-j)h)}(e^{-2\pi i B^{-1}A h} - I)(B^{-1}A)^{-1}\|  \\
&\leq&  
\|E^{-1}\| \, \|E\| \max_j \frac{|e^{-2\pi i \lambda_j  h}-1|}{|\lambda_j|} \\
&=& \kappa_E \max_j \frac{2|\sin (\pi \lambda_j  h)|} {|\lambda_j|}
\leq 2\kappa_E \pi h.
\eeas
Similarly, we have
$\|(e^{-2\pi i B^{-1}A t} - I)(B^{-1}A)^{-1}\| \leq 2\kappa_E \pi h$.
Thus when $t+lh\in(lh, (l+1)h]$, we obtain
\beas
\|\x(t+lh)\|
&\leq&  
\|e^{2\pi i B^{-1}A (t+lh)} \b(0)\|
+\sum_{j=0}^{l-1}
\|e^{2\pi i B^{-1}A (t+(l-j)h)} (e^{-2\pi i B^{-1}A t} - I) (B^{-1}A)^{-1}B^{-1}\b(j+1)\| \\
&& +\, \|(e^{-2\pi i B^{-1}A t} - I) (B^{-1}A)^{-1}B^{-1}\b(l+1)\|
 \\
&\leq &  \kappa_E \|\b(0)\|
+\sum_{j=0}^{l-1}
2\kappa_E \pi h
\|B^{-1}\| \|\b(j+1)\|
+2\kappa_E \pi h
\|B^{-1}\| \|\b(l+1)\|
\eeas
The above upper bound is an inner product between $(\kappa_E, 2\kappa_E \pi h
\|B^{-1}\|, \ldots, 2\kappa_E \pi h
\|B^{-1}\|)$ and $(\|\b(0)\|, \ldots, \|\b(l+1)\|)$. By Cauchy–Schwarz inequality and note that the norm of the second vector is smaller than 1 by assumption, we have
\[
\|\x(t+lh)\| = O(\sqrt{l} \kappa_E \|B^{-1}\| h ).
\]


Since $\x(lh) = YF|l\rangle$ and $F$ is unitary, we have
\bes
\|Y\|_F =  \|YF\|_F =
\sqrt{\sum_{l=0}^p \|YF|l\rangle\|^2}
=
\sqrt{\sum_{l=0}^p \|\x(lh)\|^2}
\leq  O(p\kappa_E \|B^{-1}\| h).
\ees
Since $ph = O(1/\epsilon)$,
it follow that
$\|M^{-1}\|=O(\kappa_E \|B^{-1}\|/\epsilon)$. Together with equation (\ref{cond:bound 1}), the condition number is upper bounded by
\[
O(\kappa_E \|B^{-1}\| \|I\otimes A-\T^{-1}D\otimes B\|/\epsilon)
=
O(\kappa_E \|B^{-1}\|
(\|A\|+\T^{-1}p\|B\| )/\epsilon),
\]
as claimed.
\end{proof}

\begin{rem} \label{remark:condition number}
{\rm Assume that $A$ is invertible. In the above proof,  when estimating $
\|\x(t+lh)\|$, 
if we bound 
\[
\|e^{2\pi i B^{-1}A (t+(l-j)h)} (e^{-2\pi i B^{-1}A t} - I) (B^{-1}A)^{-1}B^{-1} \| 
\]
by 
\[
\|e^{2\pi i B^{-1}A (t+(l-j)h)} (e^{-2\pi i B^{-1}A t} - I)\| \|(B^{-1}A)^{-1}B^{-1}\|
\leq 2\kappa_E \pi  h \|A^{-1}\|,
\]
then we have
$\|\x(t+lh)\|= O(\sqrt{l} \kappa_E \|A^{-1}\| h).$
Together with the above proof, we now have
$
\|\x(t+lh)\|= O(\sqrt{l} \kappa_E \min(\|A^{-1}\| ,\|B^{-1}\|) h ).
$
This means that the upper bounded of the condition number can be reduced to $O(\kappa_E \min(\|A^{-1}\| ,\|B^{-1}\|)
(\|A\|+\T^{-1}p\|B\| )/\epsilon)$ if $A$ is nonsingular.
}\end{rem}

At the end of this part, we use the technique of
linear combinations of unitaries (LCU, see \cite[Chapter 26]{childs2017lecture}) to construct the block-encoding of $M$.

\begin{prop}[Block-encoding of $M$]
\label{lem:block-encoding of M}
Let $M$ be the matrix  (\ref{coeff matrix}).
Given an $(\alpha_A,q_A,\delta)$-block-encoding
of $A$ in time $O(T_A)$, an $(\alpha_B,q_B,\delta)$-block-encoding
of $B$ in time $O(T_B)$, 
there is a quantum circuit that implements an $(\alpha,\max\{q_A,q_B\}+2\log p + 3,4\delta)$-block-encoding of $M$ in time $O(T_A+T_B)$,
where $\alpha 
=O(\alpha_A + \T^{-1}p \alpha_B )$.
\end{prop}

\begin{proof}
We decompose $M$ as follows
\bea \label{decomposition of M}
M &=& (F^T \otimes I) (I\otimes A)
-
(F^T \otimes I) (\tau^{-1}D\otimes B)
+
\left[ \begin{array}{ccccccc}
\langle 0| F^T \otimes I \\
0
\end{array} \right] \nonumber \\
&& -\,
\left[ \begin{array}{ccccccc}
\langle 0| F^T \otimes I \\
0
\end{array} \right](I\otimes A)
+
\left[ \begin{array}{ccccccc}
\langle 0| F^T \otimes I \\
0
\end{array} \right](\tau^{-1}D\otimes B).
\eea
The basic idea is to construct the block-encoding of each term, then use LCU to create the block-encoding of $M$.

Since $\tau^{-1}D$ is diagonal, an $((p-1)/2\tau,\log p,0)$-block-encoding can be constructed naively, for example by 2-dimensional rotations.
Note that
\[
\left[ \begin{array}{ccccccc}
\langle 0| F^T \otimes I \\
0
\end{array} \right]
=\left[ \begin{array}{ccccccc}
\langle 0| F^T  \\
0
\end{array} \right]\otimes I
\]
and $F$ is unitary, so we naturally
have a $(1,\log p,0)$-block-encoding of 
$\left[ \begin{array}{ccccccc}
\langle 0| F^T\otimes I \\
0
\end{array} \right]$.
Together with the  above two block-encodings and the block-encodings of $A,B$, we can construct the block-encodings of three other terms by 
Lemma \ref{lem:product block-encoding}.


Before we show the details of the LCU procedure, we first introduce some notation of the block-encodings. Let $U_1,U_2,U_3$ be the block-encodings of $A,B,\T^{-1}D$ respectively. Then
$V_1:=(F^T\otimes I)(I\otimes U_1)$ is an
$(\alpha_A, q_A+\log p,\delta)$-block-encoding of $(F^T\otimes I)(I\otimes A)$. Notice
that $U_3\otimes U_2$ is an $(\alpha_B(p-1)/2\tau,q_B+\log p,\delta)$-block-encoding of $\T^{-1}D\otimes B$, so $V_2:=(F^T\otimes I)(U_3\otimes U_2)$
is an $(\alpha_B(p-1)/2\tau,q_B+\log p,\delta)$-block-encoding 
of $(F^T\otimes I)(\T^{-1}D\otimes B)$. 
Let $V_3$ be an $(1,\log p,0)$-block-encoding of $\left[ \begin{array}{ccccccc}
\langle 0| F^T \otimes I \\
0
\end{array} \right]$. 
As for $\left[ \begin{array}{ccccccc}
\langle 0| F^T \otimes I \\
0
\end{array} \right](I\otimes A),
$
by Lemma \ref{lem:product block-encoding},
we can construct an $(\alpha_A,q_A+2\log p,
\delta)$-block-encoding, denoted as $V_4$.
The cost is $O( T_A)$.
Finally, by Lemma \ref{lem:product block-encoding}
again, 
we can create an $(
\alpha_B(p-1)/2\tau, q_B+2\log p , \delta
)$-block-encoding of $\left[ \begin{array}{ccccccc}
\langle 0| F^T \otimes I \\
0
\end{array} \right](D\otimes B)$
in time $O( T_B )$.
Denote this block-encoding as $V_5$.

To obtain the block-encoding of $M$, by (\ref{decomposition of M}) it suffices to apply LCU to implement
$V := 
\alpha_A V_1-(\alpha_B(p-1)/2\T)V_2 + V_3 
- \alpha_A V_4+(\alpha_B(p-1)/2\T)V_5$.
The procedure is the claimed
block-encoding. Let
$\alpha = \alpha_A + \alpha_B(p-1)/2\T + 1
+ \alpha_A + \alpha_B(p-1)/2\T$ denote the absolute sum of the coefficients, then for any state $|\phi\rangle$, the LCU produce gives a state of the form
$\alpha^{-1} |0\rangle\otimes V|\phi\rangle + |0\rangle^\bot$.
This is an $(\alpha,\max\{q_A,q_B\}+2\log p + 3,4\delta)$-block-encoding of $M$.
Here the other $3=\lceil \log 5\rceil$ ancilla qubits come from the LCU.
\end{proof}


\subsection{Main result}

With the above preliminaries, we now state
our main theorem as follows.

\begin{thm}
\label{thm:GEP by ODE-general}
Let $\epsilon \in(0,1)$, $\delta = o(\epsilon/(\kappa_M^2 \log^3 (\kappa_M/\epsilon))$, where $\kappa_M$ is the condition number of $M$.
Suppose we have an $(\alpha_A,q_A,\delta)$-block-encoding of $A$ that is constructed in time $O(T_A)$,
an $(\alpha_B,q_B,\delta)$-block-encoding of $B$ that is constructed in time $O(T_B)$.
Denote 
the generalized eigen-pairs of $A\x=\lambda B\x$ as $\{(\lambda_j,|E_j\rangle): j=1,\ldots,n\}$,
the condition number of the matrix of generalized eigenvectors as $\kappa_E$.
Suppose that $B$ is invertible,
$B^{-1}A$ is diagonalizable and 
all $\lambda_j$ are real.
Then given access to copies of the state
$\sum_j\beta_j|E_j\rangle$ that is prepared in time $O(T_{\rm in})$, there is a quantum algorithm that returns a state proportional to
$\sum_j\beta_j|\tilde{\lambda}_j\rangle|E_j\rangle$ in time
\be \label{result of main thm}
\widetilde{O}\Bigg(
\frac{ \kappa_E }{\epsilon}
\Big(
(\alpha_A + \rho \alpha_B )( T_{\rm in} + T_A+T_B)  
\Big)
\|B^{-1}\|
\Bigg),
\ee
where $|\tilde{\lambda}_j-\lambda_j|\leq \epsilon$ and $\rho = \max\{1,\max_j |\lambda_j|\}$.
\end{thm}

\begin{proof}
By (\ref{important parameters}),
$p=O(\rho/\epsilon), \T = O(1/\epsilon)$. By Proposition \ref{prop:condition number}, the condition number $\kappa_M$ of 
$M$ is bounded by
$O( \epsilon^{-1}  \kappa_E
(\|A\|+\rho \|B\| )
\|B^{-1}\|)$.
By Proposition \ref{lem:block-encoding of M}, we can construct an
$(O(\alpha_A + \rho  \alpha_B ),\max\{q_A,q_B\}+2\log p + 3,4\delta)$-block-encoding of the coefficient matrix
of the linear system (\ref{linear system}) in time 
$
O( T_A +T_B ), 
$
By Lemma \ref{lem:quantum linear solver}, the quantum state of $\VEC(C)$ can be produced in time
\bea
&& \widetilde{O}\left(\kappa_M\left(T_{\rm in} + 
\frac{(\alpha_A + \rho  \alpha_B )(T_A+T_B)}
{\|M\|} \right)\right) \nonumber \\
&=&
\widetilde{O}\left(\frac{\kappa_E  (\alpha_A + \rho  \alpha_B )(T_{\rm in}+T_A+T_B) \|B^{-1}\| }{\epsilon}
\right) . \label{overall-cost}
\eea
In the above, we used the fact that $\|A\|\leq \alpha_A, \|B\|\leq \alpha_B$.
From equations (\ref{state of C}), (\ref{state of the solution}),
when apply $F U_p F \otimes I$ to $|\VEC(C)\rangle$, we will 
obtain a state proportional to $\sum_j\alpha_j|\tilde{\lambda}_j\rangle|E_j\rangle$. The overall complexity is (\ref{overall-cost}), the cost to approximate $\ket{\VEC(C)}$.
\end{proof}


By Remark \ref{remark:condition number}, if $A$ is invertible, the complexity can be improved to
\be \label{result of main thm2}
\widetilde{O}\Bigg(
\frac{ \kappa_E }{\epsilon}
\Big(
(\alpha_A + \rho \alpha_B )( T_{\rm in} + T_A+T_B)  
\Big)
\min\Big(\|A^{-1}\|, \|B^{-1}\|\Big)
\Bigg).
\ee
Finally, for the convenience of further applications, we summarize our algorithm as follows.

\begin{breakablealgorithm}
\label{alg}
\caption{Quantum algorithm for solving QGEP}
\begin{algorithmic}[1]
\REQUIRE
(1). Two $n\times n$ complex matrices $A,B$ with eigenpairs $\{(\lambda_j,\ket{E_j}): j\in [n]\}$. Assume that $B$ is invertible, $B^{-1}A$ is diagonalizable and all $\lambda_j$ are real. \\
(2). Block-encodings of $A,B$. \\
(3). An upper bound $\rho\geq 1$ on the eigenvalues. \\
(4). Quantum access to copies of the state $|\phi\rangle$, which formally equals $\sum_{j=1}^n \beta_j |E_j\rangle$.
\\
(5). The precision $\epsilon\in(0,1)$, $h = 1/\rho$, $p=\lceil \rho/\epsilon \rceil$, $\tau = ph$.
\ENSURE The quantum state
\be
\sum_{j=1}^n \beta_j |\tilde{\lambda}_j\rangle |E_j\rangle 
\label{final state}
\ee
up to a normaliation, where $|\tilde{\lambda}_j- \lambda_j|\leq \epsilon$ for all $j$.
\STATE Construct the matrix $M$ based on equation (\ref{coeff matrix}).
\STATE Use Proposition \ref{lem:block-encoding of M} to construct the block-encoding of $M$.
\STATE Solve the linear system $M \y = \ket{0..0}\ket{\phi}$ (see (\ref{linear system})) in a quantum computer by Lemma \ref{lem:quantum linear solver}.
\STATE Return the state $(F^{-1} U_p F \otimes I) \ket{\y}$, where $U_p$ is given by equation (\ref{unitaryUp}) and $F$ is the quantum Fourier transform.
\end{algorithmic}
\end{breakablealgorithm}


\subsection{Some corollaries}

In this section, we consider some special cases of Theorem \ref{thm:GEP by ODE-general}.
First, when $A,B$ are sparse,
combining Lemma \ref{sparse block-encoding}
and Theorem \ref{thm:GEP by ODE-general},
we have the following result.

\begin{cor}[GEP of sparse matrices]
\label{thm:GEP by ODE}
Let $\epsilon \in(0,1)$.
Suppose $A,B$ are $s_A,s_B$-sparse, respectively.
Denote 
the generalized eigen-pairs of $A\x=\lambda B\x$ as $\{(\lambda_j,|E_j\rangle): j=1,\ldots,n\}$,
the condition number of the matrix of generalized eigenvectors as $\kappa_E$. Suppose that $B$ is invertible,
$B^{-1}A$ is diagonalizable and 
all $\lambda_j$ are real.
Given access to copies of the state
$\sum_j\beta_j|E_j\rangle$ that is prepared in time $O(T_{\rm in})$, there is a quantum algorithm that returns a state proportional to
$\sum_j\beta_j|\tilde{\lambda}_j\rangle|E_j\rangle$ in time
\be \label{thm2:complexity}
\widetilde{O}\Bigg(
\frac{ \kappa_E }{\epsilon}
\Big(
(s_A\|A\|_{\max} + \rho s_B\|B\|_{\max} )  T_{\rm in}
\Big)
\|B^{-1}\|
\Bigg),
\ee
where $|\tilde{\lambda}_j-\lambda_j|\leq \epsilon$ and $\rho = \max\{1,\max_j |\lambda_j|\}$.
\end{cor}

For the standard eigenvalue problem (that is $B=I$),
we have $\rho(A,I)\leq \|A\|\leq \alpha_A$.
As a corollary of Theorem \ref{thm:GEP by ODE-general},
we have the following result,
which improves the quantum algorithm proposed in
 \cite{shao2019computing}
in terms of $\kappa_E, \rho$ and $\epsilon$.

\begin{cor}[Estimating eigenvalues]
\label{cor1:SEP}
Let $\epsilon\in(0,1), \delta = o(\epsilon/(\kappa_M^2 \log^3 (\kappa_M/\epsilon))$.
Suppose we have an $(\alpha_A,q_A,\delta)$-block-encoding of $A$ that is constructed in time $O(T_A)$.
Denote 
the eigen-pairs of $A\x=\lambda \x$ as $\{(\lambda_j,|E_j\rangle): j=1,\ldots,n\}$,
the condition number of the matrix of  eigenvectors as $\kappa_E$.
Suppose $A$ is diagonalizable and 
all $\lambda_j$ are real.
Then given access to copies of the state
$\sum_j\beta_j|E_j\rangle$ that is prepared in time $O(T_{\rm in})$, there is a quantum algorithm that returns a state proportional to
$\sum_j\beta_j|\tilde{\lambda}_j\rangle|E_j\rangle$ in time
\be
\widetilde{O}\left(
\frac{\alpha_A \kappa_E ( T_{\rm in} + T_A )}{\epsilon}
\right),
\ee
where $|\tilde{\lambda}_j-\lambda_j|\leq \epsilon$ and $\rho = \max\{1,\max_j |\lambda_j|\}$.
\end{cor}

If we assume that $A$ is Hermitian, the above result is the same as that of QPE \cite{chakraborty2018power}.
Another special case is $A = I$. Now we have $B\x = \lambda^{-1} \x$. This corresponds to the estimation of the inverse of eigenvalues. If we use Corollary \ref{cor1:SEP} directly, then the complexity to estimate the inverse of eigenvalues up to additive error $\epsilon$ is 
$
\widetilde{O} (
{\alpha_B \rho_0^2 \kappa_E ( T_{\rm in} + T_B )}/{\epsilon}
),
$
where $\rho_0 = 1/\min_{j,\lambda_j\neq 0} |\lambda_j|$. However, if we consider relative error, then the cost is $
\widetilde{O} (
{\alpha_B \rho_0 \kappa_E ( T_{\rm in} + T_B )}/{\epsilon}
)$.
For estimating the inverse of eigenvalues up to additive error $\epsilon$, we can directly use Theorem \ref{thm:GEP by ODE-general} (note that $A$ is invertible now, so the cost is  (\ref{result of main thm2})). The complexity is
$
\widetilde{O} (
{\alpha_B \rho_0 \kappa_E ( T_{\rm in} + T_B )}/{\epsilon}
).
$
In conclusion, we have the following result.

\begin{cor}[Estimating inverse of eigenvalues]
\label{cor2:SEP}
Let $\epsilon\in(0,1), \delta = o(\epsilon/(\kappa_M^2 \log^3 (\kappa_M/\epsilon))$.
Suppose we have an $(\alpha_A,q_A,\delta)$-block-encoding of $A$ that is constructed in time $O(T_A)$.
Denote 
the eigen-pairs of $A\x=\lambda \x$ as $\{(\lambda_j,|E_j\rangle): j=1,\ldots,n\}$,
the condition number of the matrix of  eigenvectors as $\kappa_E$.
Suppose $A$ is invertible and diagonalizable, and 
all $\lambda_j$ are real.
Then given access to copies of the state
$\sum_j\beta_j|E_j\rangle$ that is prepared in time $O(T_{\rm in})$, there is a quantum algorithm that returns a state proportional to
$\sum_j\beta_j|\tilde{\lambda}_j^{-1}\rangle|E_j\rangle$ in time
\be
\widetilde{O}\left(
\frac{\alpha_A \rho_0 \kappa_E ( T_{\rm in} + T_A )}{\epsilon}
\right),
\ee
where $\rho_0 =  \max\{1, 1/\min_{j,\lambda_j\neq 0} |\lambda_j| \}$. The error to approximate $\lambda_j^{-1}$  can either be relative or additive.
\end{cor}

All the algorithms proposed above are also applicable to the case when the (generalized) eigenvalues are all   purely imaginary. The more general case than GEP is known as the polynomial eigenvalue problem (PEP) \cite{ruhe1973algorithms}, which aims to determine those values of $\lambda \in \mathbb{C}$ and those vectors $\x\in \mathbb{C}^n$ for which $\sum_{k=0}^m \lambda^k A_k \x = 0$, where $A_0,\ldots,A_m$ are square matrices. For solving the PEP,
the classical and most widely used approach is linearization  \cite{mackey2006vector}. That is we convert the PEP into a larger size GEP with the same eigenvalues so that the classical methods for GEPs can be pressed into service. Therefore, if the obtained GEP satisfies the assumptions of Problem \ref{prob:problem 1} (e.g., quadratic eigenvalue problems arising from applications in overdamped systems and gyroscopic systems \cite{lancaster2002lambda}), our algorithm proposed above can be applied to solve this PEP.

\section{The standard algorithm for symmetric GEPs}
\label{section:Block-encoding method}

In this section, we analyze the standard quantum algorithm that uses QPE to solve symmetric GEPs in the framework of block-encoding. Recall that a matrix pair $(A,B)$ is called symmetric if $A,B$ are Hermitian and $B$ is positive definite.  From the proof of Proposition
\ref{lemma:real generalized eigenvalues}, to solve the symmetric GEP $A\x=\lambda B\x$, it suffices to solve the Hermitian eigenvalue problem $\widetilde{A} \y = \lambda \y$, where $\widetilde{A}=B^{-1/2} A B^{-1/2}, \y = B^{1/2}\x$. 
In a quantum computer, QPE is a standard algorithm to solve Hermitian eigenvalue problems. In the framework of block-encoding (see Lemma \ref{quantum SVE}), the main problem we need to solve is  the construction of the block-encoding of $\widetilde{A}$. This can be accomplished by Lemmas \ref{negative power of block-encoding} and 
\ref{lem:product block-encoding}.
The following theorem summarizes the overall complexity of this algorithm.

\begin{thm}
\label{thm:GEP by block-encoding}
Assume that $A,B$ are $n$-by-$n$ Hermitian matrices with $B$ positive definite.
Let $U_A$ be an $(\alpha_A, q_A, \epsilon_A)$-block-encoding of $A$ that is implemented 
in time $O(T_A)$, and $U_B$ an $(\alpha_B, q_B, \epsilon_B)$-block-encoding of $B$ that is implemented 
in time $O(T_B)$.
Denote $\kappa_B$ as the condition number of $B$,
the generalized eigen-pairs of $A\x=\lambda B\x$ as $\{(\lambda_j,|E_j\rangle): j=1,\ldots,n\}$.
Let $\epsilon \in(0,1)$. Given access to copies of the state 
$\sum_j\beta_j|E_j\rangle$ that is prepared in time $O(T_{\rm in})$, there is a quantum algorithm that returns a state proportional to
$\sum_j\beta_j|\tilde{\lambda}_j\rangle|E_j\rangle$, where $|\tilde{\lambda}_j-\lambda_j|\leq \epsilon$, in time
\be \label{thm1:complexity}
\widetilde{O}\left(
\kappa_B^{0.5} T_{\rm in}  
+ \frac{\alpha_A\kappa_B^{1.5}}
{\|B\|\epsilon}
\left(
T_A
+ 
\frac{ \alpha_B \kappa_B T_B }
{  \|B\| }
\right) \right).
\ee
\end{thm}

\begin{proof}  
The basic idea of the quantum algorithm is as follows:
Denote $|\phi_0\rangle = \sum_j\beta_j|E_j\rangle$.
Note that $\widetilde{A} B^{1/2}|E_j\rangle= \lambda_j B^{1/2}|E_j\rangle$ where
$\widetilde{A} =B^{-1/2}AB^{-1/2}$.
So in step one, we apply $B^{1/2}$ to $|\phi_0\rangle $
to obtain $|\phi_1\rangle \propto \sum_j\beta_jB^{1/2}|E_j\rangle$.
Then perform QPE to $\widetilde{A}$ with initial state $|\phi_1\rangle$ to create
$|\phi_2\rangle \propto \sum_j\beta_j|\tilde{\lambda}_j\rangle \otimes B^{1/2}|E_j\rangle$. 
Finally, apply $B^{-1/2}$ to the second register
of $|\phi_2\rangle$.

To make few confusions on the notation, we will use $\gamma$, instead of $\epsilon$,
to indicate the accuracy to approximate the eigenvalues.
For an $(\alpha,q,\epsilon)$-block-encoding, 
$q$ is poly-log in other parameters like the dimension, condition number and the accuracy,
so we are more concerned about $\alpha,\epsilon$ in the time complexity analysis below.
To simplify the notation, we 
just say it an $(\alpha,\epsilon)$-block-encoding.

To implement the above three steps, we first consider the construction of the 
block-encoding of $\widetilde{A}=B^{-1/2}AB^{-1/2}$.
Denote $\widetilde{B}=B/\|B\|$, then
by Fact \ref{fact of BE}, 
$U_B$ can be viewed as an
$(\alpha_B/\|B\|,   \epsilon_B/\|B\|)$-block-encoding of $\widetilde{B}$. 
In Lemma \ref{negative power of block-encoding}, 
if we choose $\delta=\kappa_B^{1.5}\epsilon_B/\|B\|$,
then we have a
$(2\sqrt{\kappa_B}, \kappa_B^{1.5}\epsilon_B/\|B\|)$-block-encoding of $\widetilde{B}^{-1/2}=\|B\|^{1/2}
B^{-1/2}$. It can be implemented
in time 
\be\label{cost of block-encoding}
O(\alpha_B\|B\|^{-1} \kappa_B T_B  \log^{2}(\|B\|/\epsilon_B))
=\widetilde{O}(\alpha_B\|B\|^{-1} \kappa_B T_B  ).
\ee
By Fact \ref{fact of BE}, the block-encoding of $\widetilde{B}^{-1/2}$ can be viewed as a
$
(2\sqrt{\kappa_B/\|B\|},  \kappa_B^{1.5}\epsilon_B/\|B\|^{1.5})
$-block-encoding of $B^{-1/2}$.
By Lemma \ref{lem:product block-encoding},
we can construct a
$
(4\alpha_A\kappa_B/\|B\|,
\epsilon' )
$-block-encoding of $\widetilde{A}$ in time $O(T_A + \alpha_B \kappa_B T_B/\|B\|)$,
where
\be \label{BE-error}
\epsilon' = 
\frac{4\alpha_A\kappa_B^2\epsilon_B}{\|B\|^2}
+\frac{4\kappa_B\epsilon_A}{\|B\|}.
\ee

Step 1, apply $B^{1/2}$ to $|\phi_0\rangle$. 
By Lemma \ref{negative power of block-encoding},
if $\delta = \kappa_B \epsilon_B/\|B\|$, then
we can build a
$(2, \kappa_B \epsilon_B/\|B\|)$-block-encoding of $\sqrt{\widetilde{B}}=B^{1/2}/\|B\|^{1/2}$ in cost
$
\widetilde{O}(\alpha_B\|B\|^{-1} \kappa_B T_B).
$
By Fact \ref{fact of BE},
this can be viewed as an $(2\|B\|^{1/2}, \kappa_B \epsilon_B/\|B\|^{1/2})$-block-encoding of $B^{1/2}$.
Based on this block-encoding, we will obtain
\bes
|\phi_1\rangle = \frac{1}{2\|B\|^{1/2}} B^{1/2}
|\phi_0\rangle |0\rangle + |0\rangle^\bot = 
\frac{1}{2\|B\|^{1/2}}\sum_{j=1}^n \beta_j B^{1/2}|E_j\rangle |0\rangle + |0\rangle^\bot.
\ees
The error on the term $B^{1/2}
|\phi_0\rangle$ is bounded by
$\kappa_B \epsilon_B/\|B\|^{1/2}$.
The complexity of this step is
$O(T_{\rm in}+\alpha_B\|B\|^{-1} \kappa_BT_B)$.

Step 2, perform QPE to $\widetilde{A}$ based on Lemma \ref{quantum SVE}. 
Recall that we use $\gamma$ to denote the accuracy of approximating the eigenvalues,
so in Lemma \ref{quantum SVE}, we choose
\[
\tilde{\epsilon} = \frac{4\epsilon'\log^2{(1/\gamma)}}{\gamma} ,
\]
where $\epsilon'$ is given in (\ref{BE-error}).
Then from the state $|\phi_1\rangle$, we
can obtain
\bes
|\phi_2\rangle = \frac{1}{2\|B\|^{1/2} }\sum_{j=1}^n \beta_j |\tilde{\lambda}_j\rangle\otimes B^{1/2}|E_j\rangle |0\rangle + |0\rangle^\bot
\ees
in time
\be
\widetilde{O}\left(
T_{\rm in}+\alpha_B\|B\|^{-1} \kappa_B T_B
+ \frac{\alpha_A\kappa_B}{\|B\|\gamma}
\left(
T_A
+ 
\frac{ \alpha_B \kappa_B T_B }
{  \|B\| }
\right) \right).
\label{cost}
\ee

Step 3, apply $B^{-1/2}$ to the second register of $|\phi_2\rangle$, then we obtain
\beas
|\phi_3\rangle 
&=& \frac{1}{4\sqrt{\kappa_B}}
\sum_{j=1}^n \beta_j |\tilde{\lambda}_j\rangle |E_j\rangle |0\rangle + |0\rangle^\bot.
\eeas
The error on the summation term is bounded by
$\kappa_B^{2.5}\epsilon_B^2/\|B\|^{2}$.
The complexity of this step is
(\ref{cost of block-encoding})
+ (\ref{cost}).

We can choose $\epsilon_B$
small to make sure the error on the state
$|\phi_3\rangle$ is small.
Since $\epsilon_B$ appears as a poly-log term
in the complexity, this does not affect the
complexity too much. By
amplitude estimation, it costs an extra
$O(\sqrt{\kappa_B})$ to obtain
the state proportional to $\sum_{j=1}^n \beta_j |\tilde{\lambda}_j\rangle |E_j\rangle$.
Multiplying (\ref{cost}) by $O(\sqrt{\kappa_B})$
gives rise to the claimed result.
\end{proof}

The complexity result of Theorem \ref{thm:GEP by block-encoding} is invariant under scaling, so we can assume that $\|B\| = \Theta(1)$. Then the complexity can be simplified into $\widetilde{O}(
\kappa_B^{0.5} T_{\rm in}  
+ {\alpha_A\kappa_B^{1.5}}
(T_A
+ 
\alpha_B \kappa_B T_B 
)/\epsilon )$.
Specifically, when $A,B$ are sparse, 
as a corollary of Lemma \ref{sparse block-encoding} and Theorem \ref{thm:GEP by block-encoding},
we have the following result.

\begin{cor}
\label{cor:GEP by block-encoding}
Assume that $A,B$ are $n$-by-$n$ Hermitian matrices and $B$ is positive definite.
Assume that $A$ (resp. $B$) has sparsity $s_A$ (resp. $s_B$). Denote $\kappa_B$ as the condition number of $B$,
the generalized eigen-pairs of $A\x=\lambda B\x$ as $\{(\lambda_j,|E_j\rangle): j=1,\ldots,n\}$.
Let $\epsilon \in (0,1)$.
Given access to copies of the state
$\sum_j\beta_j|E_j\rangle$ that is prepared in time $O(T_{\rm in})$, then there is a quantum algorithm that returns a state proportional to
$\sum_j\beta_j|\tilde{\lambda}_j\rangle|E_j\rangle$ in time
\be
\widetilde{O}\left(
\kappa_B^{0.5} T_{\rm in}
+ \frac{\kappa_B^{2.5} s_A s_B  \|A\|_{\max} \|B\|_{\max}}
{\epsilon\|B\|^2}
\right),
\ee
where $|\tilde{\lambda}_j-\lambda_j|\leq \epsilon$.
\end{cor}

\begin{proof}
When $A,B$ are sparse, by Lemma \ref{sparse block-encoding}, we can construct an
$(s_A\|A\|_{\max},{\rm poly}\log(n/\epsilon),\|A\|_{\max}\epsilon)$-block-encoding of $A$,
an $(s_B\|B\|_{\max},{\rm poly}\log(n/\epsilon),\|B\|_{\max}\epsilon)$-block-encoding of $B$ efficiently.
Then the claimed result comes
naturally from
Theorem \ref{thm:GEP by block-encoding}.
\end{proof}

As we can see from the proof of 
Theorem \ref{thm:GEP by block-encoding}, the main cost comes from the construction of the block-encoding
of $\widetilde{A}$. In \cite{chakraborty2018power}, there is another
way to create the block-encoding
of the product of two matrices.

\begin{lem}[Lemma 5 of \cite{chakraborty2018power}]
\label{product block-encoding 1}
Let $A_1,A_2$ be two matrices.
Assume that $\|A_i\|\leq 1, 
\alpha_i \geq 1$ 
for $i=1,2$.
Let $U_i$ be an $(\alpha_i, q_i, \epsilon_i)$-block-encoding of $A_i$
that can be implemented in time $O(T_i)$. Then there is a 
$(2,q_1+q_2+2, \sqrt{2}(\epsilon_1+\epsilon_2))$-block-encoding of $A_1A_2$ that
can be implemented in time
$O(
\alpha_1(q_1+T_1) +
\alpha_2(q_2+T_2)
)$.

\end{lem}

As a corollary of Fact \ref{fact of BE} and Lemma \ref{product block-encoding 1},
when the assumption $\|A_i\|\leq 1$ is
removed, 
we then have the following result.

\begin{cor}
\label{product block-encoding-2}
Let $U_i$ be an $(\alpha_i,q_i,\epsilon_i)$-block-encoding of $A_i$
that is obtained in time $O(T_i)$, where $i=1,2$. Then we can construct
a
$(2\|A_1\|\|A_2\|,q_1 +q_2 +2, \sqrt{2}(\|A_1\|\epsilon_2+\|A_2\|\epsilon_1))$-block-encoding of $A_1A_2$ in time
$O(
\alpha_1(q_1+T_1)\|A_1\|^{-1} +
\alpha_2(q_2+T_2)\|A_2\|^{-1} 
)$.

\end{cor}

\begin{proof}
Without loss of generality, we assume that $\|A_i\|\geq 1$.
Denote $A_i' = A_i/\|A_i\|$. Then
$U_i$ can be viewed as an $(\alpha_i/\|A_i\|, q_i, \epsilon_i/\|A_i\|)$-block-encoding of $A_i'$. Since $\|A_i\|\geq 1$,
we have $\alpha_i\geq \|A_i\|$. This
implies $\|A_i'\|\leq 1$.
By  Lemma \ref{product block-encoding 1},
a $(2,q_1+q_2+2, \sqrt{2}(\epsilon_1\|A_1\|^{-1}+\epsilon_2\|A_2\|^{-1}))$-block-encoding of $A_1'A_2'$ 
can be constructed in time
$O(
\alpha_1(q_1+T_1)\|A_1\|^{-1} +
\alpha_2(q_2+T_2)\|A_2\|^{-1} 
)$.
By Fact \ref{fact of BE}, this is a
$(2\|A_1\|\|A_2\|,q_1+q_2+2, \sqrt{2}\|A_1\|\|A_2\|(\epsilon_1\|A_1\|^{-1}+\epsilon_2\|A_2\|^{-1}))$-block-encoding of $A_1A_2$.
\end{proof}

With the same notation as that in the proof of Theorem \ref{thm:GEP by block-encoding},
by Corollary \ref{product block-encoding-2}
we can construct an
$
(2\|B^{-1/2}\| \|AB^{-1/2}\|,
\epsilon' )
$-block-encoding of $\widetilde{A}$ in time
\be
\widetilde{O}\left(
\frac{ \alpha_A T_A\|B^{-1/2}\|}{\|A B^{-1/2}\|}
+ 
\frac{\|A\|\kappa_B^{1.5}\alpha_B T_B}
{ \|A B^{-1/2}\| \|B\|^{1.5}}
\right)
=\widetilde{O}\left(
\frac{ \kappa_B^{0.5}\alpha_A T_A}{\|A\|}
+ 
\frac{ \kappa_B^{1.5}\alpha_B T_B }
{  \|B\| }
\right),
\ee
where
\be \label{BE-error}
\epsilon' = 
\sqrt{2} 
\left(\sqrt{2} \|B^{-1/2}\|
\left(\|A\|\frac{\kappa_B^{1.5} \epsilon_B}{\|B\|^{1.5}}
+ \|B^{-1/2}\| \epsilon_A
\right)
+ \|AB^{-1/2}\|\frac{\kappa_B^{1.5} \epsilon_B}{\|B\|^{1.5}}
\right).
\ee
Now with a similar arguments to the proof of Theorem \ref{thm:GEP by block-encoding} (the only difference is the parameters in the block-encoding of $\widetilde{A}$), it is not hard to show that
the complexity of Theorem \ref{thm:GEP by block-encoding} becomes
\be
\widetilde{O}\left(
\kappa_B^{0.5} T_{\rm in}
+
\frac{\kappa_B^2}{\|B\|\epsilon}
\left(\alpha_A T_A
+ \frac{ \kappa_B \alpha_B T_B  \|A\| }
{\|B\| }\right)
\right) .
\ee
And the complexity of Corollary \ref{cor:GEP by block-encoding} becomes
\be \label{com:GEP by block-encoding2}
\widetilde{O}\left(
\kappa_B^{0.5} T_{\rm in}
+
\frac{\kappa_B^2}{\|B\|\epsilon}
\left(s_A\|A\|_{\max}  
+ \frac{ \kappa_B s_B  \|A\| \|B\|_{\max}}
{\|B\| }\right)
\right) .
\ee
In comparison, the above results have better dependence on $\alpha_A,\alpha_B$ or
$s_A\|A\|_{\max},s_B\|B\|_{\max}$ but a little worse dependence on the condition number $\kappa_B$.

Finally, we make a comparison between our algorithm given in the above section and the above standard algorithm based on QPE for symmetric GEPs. For simplicity, we assume that $\|B\|=\Theta(1)$ and $T_{\rm in}= T_A=T_B=\widetilde{O}(1)$. Then the complexity of our algorithm is 
\[
\widetilde{O}\left(
(\alpha_A + \|A\|\kappa_B \alpha_B  ) \kappa_B^{1.5}  /\epsilon
\right).
\]
And the complexity of the standard quantum algorithm by QPE is 
\[
\min\left\{
\widetilde{O}\big(
\alpha_A\alpha_B\kappa_B^{2.5} /\epsilon 
\big), ~~
\widetilde{O}\big(
(\alpha_A+\|A\|\kappa_B\alpha_B)\kappa_B^{2} /\epsilon 
\big)
\right\}.
\]
As we can see, our algorithm is a little better than the standard one in terms of $\alpha_A, \alpha_B$ and $\kappa_B$. When restricted to the Hermitian eigenvalue problem, all the algorithms have the same complexity $\widetilde{O}(\alpha/\epsilon)$, which coincides with the complexity of QPE \cite{chakraborty2018power}. Thus all the algorithms achieve optimal dependence on $\epsilon$.

Indeed, we believe that all the above algorithms for symmetric GEPs should have similar performance in practice. The main difficulty for the  QPE method is that we need to construct the block-encoding of $\widetilde{A}$ efficiently. In the proof, we obtain a $
(4\alpha_A\kappa_B/\|B\|,
\epsilon' )
$-block-encoding of $\widetilde{A}$. So in the QPE, $4\alpha_A\kappa_B/\|B\| = O(\alpha_A\kappa_B)$ is viewed as an upper bound of the eigenvalues. However, in our algorithm, we only need to choose $\|B^{-1/2}AB^{-1/2}\|=O(\|A\|\kappa_B)$ as the upper bound. If $\alpha_A \approx \|A\|$, then both algorithms have the same complexity.
In the above second method based on Corollary \ref{product block-encoding-2}, we obtain a
$
(2\|B^{-1/2}\| \|AB^{-1/2}\|,
\epsilon' )
$-block-encoding of $\widetilde{A}$. The bound $\|B^{-1/2}\| \|AB^{-1/2}\|$ is close to the one we used. While this idea leads to a little worse dependence on the condition number of $B$. Since $\widetilde{A}=B^{-1/2}AB^{-1/2}$,  the optimal block-encoding we can construct is $(\|B^{-1/2}AB^{-1/2}\|, \epsilon'')$ for some $\epsilon''$. From this point, our algorithm combines the advantages of both methods introduced above based on QPE. In other words, it has similar performance to the QPE method using the optimal block-encoding of $\widetilde{A}$.



\section{Lower bound for the generalized eigenvalue problems}
\label{section:ow bounds for the generalized eigenvalue problem}

For a classical computer,
the GEPs involving singular matrices
are challenging to solve from the viewpoint of stability and computational complexity \cite{van1979computation,hochstenbach2019solving}. For instance, in \cite[Section 7.7]{Golub}, three simple examples are listed to illustrate this by showing that when the matrices are singular, the eigenvalues can be empty, finite or the whole space.
In this section, we shall prove that
singular GEPs are also hard to solve for a quantum computer.

For any $n\times n$ matrix $M=(M_{ij})$, we define the entry access oracle as: for any $i,j\in[n]$, the oracle performs $|i,j\rangle|0\rangle \mapsto |i,j\rangle|M_{ij}\rangle$.
We consider the following simple version of the GEP. 
Assume that we are given access to copies of a quantum state $\ket{\psi} \in \mathbb{C}^n$ and query access to the entries of $n \times n$ matrices $A$ and $B$. 
We are promised that
$A \ket{\psi} = \lambda B \ket{\psi} $
for some $\lambda \in \mathbb{C}$. Our goal is to find $\lambda$ efficiently (i.e., in time $\poly\log(n)$, not $\poly(n)$).
However, we can show that there is no such efficient quantum algorithm when $B$ is singular.

\begin{thm}\label{thm:lower bound of GEP}
Suppose we are given access to copies of a quantum state $\ket{\psi}$ and query access to entries of $n \times n$ matrices $A$ and $B$.
If $A \ket{\psi} = \lambda B \ket{\psi}$ and
$B$ is singular, then
finding $\lambda$ requires at least $\Omega(\sqrt{n})$ quantum queries to $A$ and $B$.
\end{thm}

\begin{proof}
Consider the following  two cases:
\begin{enumerate}
\item $\ket{\psi} = \frac{1}{\sqrt{n}} \sum_{i=1}^n \ket{i}$, $A = B = \proj{x}$, for arbitrary $x \in \{1,\dots,n\}$;
\item $\ket{\psi} = \frac{1}{\sqrt{n}} \sum_{i=1}^n \ket{i}$, $A=0$, $B = \proj{x}$, for arbitrary $x \in \{1,\dots,n\}$.
\end{enumerate}
These are both instances of the GEP. In the first case, $A\ket{\psi} = B\ket{\psi}$, so $\lambda = 1$. In the second case, $A \ket{\psi} = 0$, so $\lambda = 0$. 
However, we can show that distinguishing these two cases requires $\Omega(\sqrt{n})$ queries to $A$ and $B$. 
We can think of this problem as follows: we are given query access to a Boolean function $f:[2n] \to \{0,1\}$, such that either:
\begin{enumerate}
    \item $f(x) = 1$ and $f(x+n) = 1$ if $x=x_0$, for some $x_0 \in [n]$, and otherwise $f(x) = 0$;
    \item $f(x) = 0$ for all $x \in [n]$; $f(x+n) = 1$ if $x=x_0$, for some $x_0 \in [n]$, and otherwise $f(x) = 0$.
\end{enumerate}
Our task is to distinguish these two cases. 
But one can prove hardness of this using the optimality of Grover's search algorithm. More precisely, define 
\[
g(x) = \frac{1}{2} \left( 1- (-1)^{f(x) + f(x+n)} \right)
\]
as a function from $[n]$ to $\{0,1\}$.
If $f$ belongs to the first class, then $g(x) = 0$ for all $x\in [n]$.
If $f$ belongs to the second class, then $g(x) = 1$ if $x = x_0$ and 0 otherwise. This is a searching problem. By the optimality of Grover's algorithm \cite{bennett1997strengths}, 
distinguishing these two types of functions requires at least $\Omega(\sqrt{n})$ queries.
\end{proof}

Under the assumption of Theorem \ref{thm:lower bound of GEP}, we can use quantum amplitude estimation \cite{brassard2002quantum} to estimate $\lambda$ up to some error. However, this kind of eigenvalue problem (i.e., the eigenvector is given to us) rarely occurs in practice, so we will not go deeper here.

\section*{Acknowledgement}

We would like to thank Ashley Montanaro for helpful discussions and suggesting the idea to prove the lower bound in Section \ref{section:ow bounds for the generalized eigenvalue problem}. 
We also would like to thank Dominic Verdon for useful comments on a previous version. 
CS was supported by the QuantERA ERA-NET Cofund in Quantum Technologies implemented within the European
Union's Horizon 2020 Programme (QuantAlgo project), and EPSRC grants EP/L021005/1 and EP/R043957/1.
JPL was supported by the National Science Foundation (CCF-1813814), the U.S. Department of Energy, Office of Science, Office of Advanced Scientific Computing Research, Quantum Algorithms Teams and Accelerated Research in Quantum Computing programs. 
No new data were created during this study.

\appendix

\section{Error analysis}
\label{appendix:Error analysis}

The analysis below is similar to that in quantum phase estimation.
Recall from (\ref{superposition solution}) that the state we have is proportional to
\bes
\frac{1}{\sqrt{p}}
\sum_{j=1}^n \beta_j \sum_{l=1}^{p}e^{2\pi i\lambda_j lh}|l\rangle |E_j\rangle.
\ees
Applying inverse of quantum Fourier transform to $\ket{l}$ yields
\be \label{app1:eq1}
\frac{1}{p}
\sum_{j=1}^n \beta_j \sum_{k=1}^{p} \left(\sum_{l=1}^{p}e^{2\pi i l(\lambda_j h- \frac{k}{p})} \right)
|k\rangle |E_j\rangle.
\ee
Assume that $k$ satisfies $|\lambda_j h- \frac{k}{p}|\leq 1/2p$, then 
\bes
\frac{1}{p} 
\left|\sum_{l=1}^{p}e^{2\pi i l(\lambda_j h- \frac{k}{p})} \right|
= \frac{1}{p} \frac{|e^{2\pi i (\lambda_j h- \frac{k}{p})p} - 1|}{|e^{2\pi i l(\lambda_j h- \frac{k}{p})} - 1|} 
= \frac{1}{p} \frac{|\sin{[(\lambda_j h- \frac{k}{p})p\pi]  }|}{|\sin{[(\lambda_j h- \frac{k}{p}) \pi ] }|} 
\geq \frac{2}{\pi}.
\ees
In the last step, we used the fact that $|\sin x| \geq 2x/\pi$ if $|x| \leq \pi/2$. Thus the 
state (\ref{app1:eq1}) can be rewritten as
\be \label{app1:eq2}
\sum_{j=1}^n \beta_j \gamma_j |k_j\rangle |E_j\rangle + {\rm others},
\ee
where $k_j$ is the integer such that $|\lambda_j h- \frac{k_j}{p}|\leq 1/2p$. It is shown above that $|\gamma_j| \geq 2/\pi$. Indeed, similar to the technique used in QPE \cite[Chapter 5.2]{NielsenChuang}, we can add more qubits to increase $|\gamma_j| \approx 1-\delta$ for any arbitrary small $\delta$. In this case, the state (\ref{app1:eq2}) is the one we aim to prepare.

From the above analysis, we see that $p$ is determined by the error to approximate $\lambda_j h$, so we should choose $\tau=ph=1/\epsilon$. Since $k/p \leq 1$, we need to choose $h$ such that $|\lambda_j h| \leq 1$, i.e., $h = 1/\rho$. Finally, we have $p = \rho/\epsilon$. This gives the choices  (\ref{important parameters}).

\section{Complement to the proof of Proposition \ref{prop:condition number}}
\label{appendix:Supplementary proof}

In this part, we show that 
the estimations in Proposition \ref{prop:condition number} are still true even if $A$ is singular.
In (\ref{solution of ODE:cond}), if $A$ is not invertible, we assume that $B^{-1}A = E \Lambda E^{-1}$, where $\Lambda = {\rm diag}(\Lambda_1,0)$. Then
\beas
\x(t+lh) 
&=& e^{2\pi i B^{-1}A (t+lh)} \b(0) 
+
2\pi i E e^{2\pi i \Lambda (t+lh)} 
\sum_{j=0}^{l-1}
\int_{jh}^{(j+1)h} 
e^{-2\pi i \Lambda s}
E^{-1}B^{-1} \b(j+1) ds \\
&& +\, 2\pi i E e^{2\pi i \Lambda (t+lh)}\int_{lh}^{t+lh}  
e^{-2\pi i \Lambda s} E^{-1}
B^{-1} \b(l+1) ds \\
&=& e^{2\pi i B^{-1}A (t+lh)} \b(0) 
-
\sum_{j=0}^{l-1} E  
\begin{bmatrix}
\Lambda_1^{-1}e^{2\pi i \Lambda_1 (t+(l-j)h)}(
e^{-2\pi i \Lambda_1 h} - I)  & \\
& (h) I 
\end{bmatrix}
E^{-1}B^{-1} \b(j+1)  \\
&& - \, E  
\begin{bmatrix}
\Lambda_1^{-1}e^{2\pi i \Lambda_1 t}(
e^{-2\pi i \Lambda_1  t} - I)  & \\
& t I 
\end{bmatrix}
E^{-1}B^{-1} \b(l+1).
\eeas
Note that $|t|\leq h$ and
\beas
\left\| E  
\begin{bmatrix}
\Lambda_1^{-1}e^{2\pi i \Lambda_1 (t+(l-j)h)}(
e^{-2\pi i \Lambda_1 h} - I)  & \\
&  hI 
\end{bmatrix}
E^{-1}  \right\| 
&\leq& \kappa_E \max\{h, \max_j \frac{|e^{-2\pi i \lambda_j  h}-1|}{|\lambda_j|} \} \\
&\leq& 2\pi  \kappa_E h.
\eeas
So we still have
$
\|\x(t+lh)\| = O(\sqrt{l} \kappa_E \|B^{-1}\| h).
$



\bibliographystyle{plain}
\bibliography{nonlinear-eigen}

\end{document}